\documentclass[a4paper,10pt]{article}
\usepackage{amssymb,amsbsy,amsmath,amsfonts,amscd,amsthm}

\usepackage{moreverb,url}

\usepackage[colorlinks,bookmarksopen,bookmarksnumbered,citecolor=red,urlcolor=red]{hyperref}
\usepackage[T1]{fontenc}
\usepackage{textcomp}
\usepackage{graphicx}
\usepackage{color}
\usepackage{pdfsync}
\usepackage{enumitem}
\usepackage[useregional]{datetime2}
\usepackage{pgf,tikz}
\usetikzlibrary{arrows}
\usetikzlibrary[patterns]

\newtheorem{theorem}{Theorem}[section]
\newtheorem{lemma}[theorem]{Lemma}
\newtheorem{proposition}[theorem]{Proposition}

\theoremstyle{remark}
\newtheorem{remark}[theorem]{Remark}
\newtheorem{example}[theorem]{Example}

\def\R{{\mathbb R}}

\def\M{{\mathbb M}}

\def\diag{\mathop{\textnormal{diag}}}
\def\vect{\mathop{\textnormal{vect}}}
\def\cof{\mathop{\textnormal{cof}}}
\def\tr{\mathop{\textnormal{tr}}}
\def\div{\mathop{\textnormal{div}}\nolimits}
\def\dev{\mathop{\textnormal{dev}}}

\def\SO{\textnormal{SO}}

\def\Sl{\textnormal{sl}}
\def\Sym{\mathop{\textnormal{Sym}}\nolimits}
\def\Skew{\mathop{\textnormal{Skew}}\nolimits}
\def\sym{\mathop{\textnormal{sym}}}
\def\skew{\mathop{\textnormal{skew}}}

\def\nx{\nabla_{\!x}}

\def\bigpenta{\begin{tikzpicture}[scale=0.15,line cap=round,line join=miter,line width=.15mm]
\draw  (2.,0.)-- (3.,0.);
\draw (3.,0.)-- (3.3090169943749475,0.9510565162951532);
\draw  (3.3090169943749475,0.9510565162951532)-- (2.5,1.5388417685876266);
\draw  (2.5,1.5388417685876266)-- (1.6909830056250525,0.9510565162951536);
\draw  (1.6909830056250525,0.9510565162951536)-- (2.,0.);
\end{tikzpicture}}
\def\pentagon{\begin{tikzpicture}[scale=0.08,line cap=round,line join=miter,line width=.08mm]%version pour superposer avec *
\draw  (2.,0.)-- (3.,0.);
\draw (3.,0.)-- (3.3090169943749475,0.9510565162951532);
\draw  (3.3090169943749475,0.9510565162951532)-- (2.5,1.5388417685876266);
\draw  (2.5,1.5388417685876266)-- (1.6909830056250525,0.9510565162951536);
\draw  (1.6909830056250525,0.9510565162951536)-- (2.,0.);
\end{tikzpicture}}
\def\penta{\hbox to 0pt{\hss\pentagon\hss}}%version pour superposer sans * (mystère)

\def\revision#1{{#1}}
\def\squareabove{\lower 4pt\hbox{\scalebox{0.5}{$\square$}}}

\title{Thermodynamics of MacMillan's liquid crystal model}

\author{Herv\'e Le Dret$^1$ and Annie Raoult$^2$}{}
\date{\today, \DTMcurrenttime}
\begin{document}
\maketitle

\medskip

{\footnotesize
 \centerline{$^1$Sorbonne Universit\'e, Universit\'e Paris Cit\'e, CNRS, Laboratoire Jacques-Louis Lions, F-75005 Paris, France}
} 

\medskip

{\footnotesize
 \centerline{$^2$Universit\'e Paris Cit\'e, CNRS, MAP5, F-75006 Paris, France}
}

\begin{abstract}
We study liquid crystal models with bulk free energy from the point of view of the second law of thermodynamics. We formulate these models as objective internal variable models. Examples of application are given for the de Gennes free energy. 
\end{abstract}

\maketitle
\section{Introduction}
We are interested in the thermodynamic consistency of certain liquid crystal models.
In addition to standard thermodynamical variables, continuum models for liquid crystals feature quantities called order parameters that reflect the microscopic orientation of the liquid crystal's constituent molecules. These quantities, either in tensor or director vector form, are the defining trait of liquid crystal models.

We work in the general direction initiated by Ericksen \cite{Ericksen59,Ericksen61}, see also \cite{Leslie}, \cite{Eringen} among many others. We are especially interested in work by MacMillan \cite{McM these}, who developed in particular an internal variable theory for liquid crystals, in which the internal variable is the aforementioned order tensor. This theory fits quite well within the approach to thermomechanics with (or without) internal variables that we developed in \cite{HLDAR1}. It should be noted that MacMillan's internal variable model actually takes a minor place in his work, since he gave preference to what he called a field model, in which the order tensor obeys a second order in time differential equation. There seems to be little justification for this kind of inertial model, so we concentrate on the internal variable model with a first order in time differential equation for the order tensor. Our main goal is to give necessary and sufficient conditions for MacMillan's internal variable model to satisfy the second law of thermodynamics, although at times, we have to content ourselves with just sufficient conditions. A secondary goal is to slightly generalize MacMillan's model  in order to give it more flexibility in terms of possible behaviors, while still ensuring that the second law is satisfied. 

The liquid crystal free energies we consider are special cases of Landau-de Gennes bulk energies, see \cite{JMB} as a general reference for the stationary case, in which there is no gradient of the order tensor. In particular, we do not consider such energies as the Oseen-Frank elastic energies, which we plan to address in future work. A word of warning: we use notation that is consistent with the scheme we introduced in  \cite{HLDAR1}, but that is not traditional in liquid crystal modeling, see Sections \ref{section notation} and \ref{section des origines} below.

After explaining our notational scheme in Section \ref{section notation}, we give a brief account of thermomechanics in Section~\ref{section thermo}. Next in Section \ref{section CN}, we perform the Coleman-Noll procedure for a thermomechanical model with an internal variable expressed in the Eulerian description. The specific nature of the internal variable is immaterial here, but we will later apply the results to the case of the order tensor of liquid crystals. These general results are of course nothing new. However, we feel that this crucial step, which determines which constitutive laws will lead to models that satisfy the second law of thermodynamics, is quite often performed with insufficient care in the literature. 

We now switch in Section \ref{section des origines} to liquid crystals properly speaking, by giving succinctly the origin of the de Gennes $Q$-tensor (which we hereafter call $\xi$) and presenting MacMillan's internal variable model. In this model, the stress tensor and internal variable flow rule are assumed to be basically affine with respect to the stretching tensor $d$, and otherwise arbitrary with respect to the temperature $\theta$ and order tensor $\xi$. As mentioned above, the Helmholtz free energy is a bulk energy, that is to say a function of $\theta$ and $\xi$,  without $\nabla\xi$ dependence, a typical example of which is the de Gennes energy.

Sections \ref{section obj}  and \ref{section traceless} are devoted to questions of frame-indifference or objectivity of the constitutive laws for the free energy, the heat flux, the Cauchy stress tensor and the order tensor together with its flow rule, harking back to the pioneering work of Rivlin-Ericksen \cite{Rivlin-Ericksen} for objective constitutive laws and of Zaremba \cite{Zaremba} for objective derivatives. Concerning objective derivatives, we present specific developments taking into account the traceless character of the order tensor.

In Section \ref{section McM compatible}, we study in depth the compatibility of MacMillan's internal variable model with the second law of thermodynamics. Let us detail this section a little bit more. In Proposition \ref{les conditions pour McM}, we first give a set of rather intricate necessary and sufficient conditions on the various constitutive laws that ensures that MacMillan's model can indeed be made compatible with the second law. Some of these conditions were already identified as necessary by MacMillan. We then study the conditions obtained in Proposition \ref{les conditions pour McM} one after the other, in order to provide conditions for the constitutive laws that are workable and as explicit as could achieve. This is the object of Propositions \ref{premiere condition k0}, \ref{une forme explicite non mais !}, \ref{sufficient form 1} and \ref{sigma0*}, which take various points of view, either focusing on Rivlin-Ericksen type representation formulas or on more direct arguments. We obtain in the end a set of sometimes necessary and sufficient, or only sufficient conditions, with or without simplifying assumptions. A simple example of appropriate constitutive laws is given for the case of the de Gennes energy.
 
 Finally, in Section \ref{Section trois pas en avant}, we propose a generalization of MacMillan's model by allowing laws that are quadratic with respect to $d$, however not in maximal generality for simplicity. In the same spirit as before, we obtain necessary and sufficient conditions in Proposition \ref{plus fort que McM...+}. In Proposition \ref{bon ok ça marche}, we show that these conditions can be satisfied and  conclude by again providing a simple example in the case of the de Gennes energy.%\Notesherve{L'exemple n'est pas dans la prop 9.2}\Notesannie{C'est mieux ainsi}
\section{Notation}\label{section notation}
We denote the set of $3\times 3$ matrices by $\M_3$, endowed with the Frobenius inner product $F:G=\tr(F^TG)$. We let  $\Sym_3$ be the set of symmetric matrices, $\Skew_3$ the set of skew-symmetric matrices, $\Sl(3)$ the set of trace-free matrices and $\SO(3)$ the set of rotation matrices. For any $M\in\M_3$, $\sym(M)$, $\skew(M)$ and $\dev (M)$ respectively denote the symmetric, skew-symmetric and deviatoric parts of $M$.

As a rule, we denote Eulerian quantities with lowercase letters. 
We thus let $v(x,t)$ for the Eulerian velocity at space-time point $(x,t)$ and $h(x,t)=\nx v(x,t)$ for its gradient. We also let $d(x,t)=\sym(h(x,t))$ for the stretching tensor and $w(x,t)=\skew(h(x,t))$ for the spin tensor. Incompressibility is expressed by $\div_x v(x,t)=0$, \emph{i.e.}, $h(x,t)\in\Sl(3)$. The temperature field is denoted $\theta(x,t)$ and its gradient $g(x,t)=\nx\theta(x,t)$. Since we consider a homogeneous, incompressible material, the mass density $\rho$ remains constant, and $h$, $\theta$ and $g$ will be our main thermodynamic variables. These are complemented by an internal variable $\xi$ with values in some vector space $V$. A specific choice for $\xi$ and $V$ that is adapted to liquid crystal modeling will be introduced later on. We keep it general for the moment.

The material derivative of any scalar-, vector- or tensor-valued Eulerian quantity $z$ is given by  $\dot z=\frac{\partial z}{\partial t}+v_i\frac{\partial z}{\partial x_i}$. In particular, the material derivative $\dot\xi$ of the internal variable is thus also $V$-valued. 

We denote the  Cauchy stress tensor by $\sigma(x,t) $ and the heat flux vector by $q(x,t)$. Some liquid crystal theories are micropolar, but we consider here $\sigma$ to be $\Sym_3$-valued.  In terms of thermodynamic state functions, we use the Helmholtz free energy specific density $a_m$, the internal energy specific density $e_m$ and the entropy specific density $s_m$. 

We make a general local state hypothesis and distinguish between a given quantity and a constitutive law for that same quantity
by using a hat for the latter, \emph{e.g.} $a_m$  for the Helmholtz free density itself as
opposed to $\widehat a_m$ for a constitutive law for it, in the sense that 
$$
    a_m(x,t)=\widehat a_m\bigl(h(x,t),\theta(x,t),g(x,t),\xi(x,t)\bigr)\text{ with }\widehat a_m\colon\Sl(3)\times\R_+^*\times\R^3\times V\to\R \text{ given,}
$$
or for short
$$
    a_m(x,t)=\widehat a_m\bigl((h,\theta,g,\xi)(x,t)\bigr),
$$
and so on.
Following Truesdell's equipresence principle \cite{Truesdell-Noll}, we assume that all constitutive laws a priori take the whole set of thermodynamic variables $(h,\theta,g,\xi)$ as arguments.

\section{Thermomechanical background}\label{section thermo}
Let us give a brief rundown of the standard thermomechanical setup, see \cite{Gurtinandco}  for instance among many other references.
The first basic equation satisfied by any material is the dynamics equation 
\begin{equation}\label{dynamics eqn}
    \rho\dot v-\div_x\sigma=b,
\end{equation}
where $b$ is the external body force density, together with appropriate initial and boundary conditions. We will not be concerned with such conditions here.

Let us then  recall the local differential form of the first law of thermodynamics, which just says that the time derivative of the total energy is equal to the sum of all power sources, either mechanical or thermal. In terms of the specific density for the internal energy, this reads
\begin{equation}\label{1st law}
    \rho\dot e_m=\sigma:h-\div_xq+r,
\end{equation}
where $r$ is the external thermal power source. Note that the dynamics equation was used to take into account the kinetic energy and remove the power due to the external body force $b$ from equation \eqref{1st law} for the internal energy. We will also refer to the first law as the energy equation.
It is assumed that both $b$ and $r$ can  in principle be set arbitrarily, even though not necessarily in practical terms. Note that since $\sigma$ is symmetric, we have $\sigma:h=\sigma:d$, \emph{i.e.}, the spin tensor does not produce any power.

The second law of thermodynamics or Clausius-Duhem inequality assumes that there exists a quantity $s$ called entropy that varies faster than its sources, which consist of thermal power terms divided by the temperature. In terms of specific densities, this reads
\begin{equation}\label{2nd law}
    \rho\dot s_m\ge -\div_x\Bigl(\frac{q}\theta\Bigr)+\frac{r}\theta.
\end{equation}
where for simplicity, we do not include any extra entropy flux, see \cite{Muller} for the introduction of an extra entropy flux. It is customary to introduce the internal dissipation 
\begin{equation}\label{dissipation interne}
    d_{\textnormal{int}}=\rho\theta \dot s_m+\div_xq-r,
\end{equation}
so that the evolution of the entropy is now, sort of tautologically, expressed as 
\begin{equation}\label{entropy balance}
    \rho\dot s_m= \frac{d_{\textnormal{int}}-\div_xq+r}\theta,
\end{equation}
and the internal dissipation appears as an internal power source. More importantly, the Clausius-Duhem inequality can be rewritten as 
\begin{equation}\label{2nd law 2}
    d_{\textnormal{int}}-\frac{q\cdot\nx\theta}{\theta}\ge 0.
\end{equation}

The Helmholtz free energy, with specific density $a_m=e_m-\theta s_m$, is especially well-suited to working with the second law, since, on combining it with the first law \eqref{1st law}, we obtain the alternate expression for the internal dissipation,
\begin{equation}\label{dissipation interne 2}
    d_{\textnormal{int}}=-\rho(\dot a_m+\dot \theta  s_m)+\sigma:d.
\end{equation}
This makes it clearly internal, since the external thermal power source $r$ and the heat flux $q$ no longer appear explicitly, whereas the internal mechanical power term $\sigma:d$ comes in. Conversely, assuming \eqref{dissipation interne 2} and \eqref{1st law}, we naturally recover \eqref{dissipation interne}.

The Clausius-Duhem inequality is often broken into two independent inequalities called the Clausius-Planck inequalities,
\begin{equation}\label{CP mech}
    d_{\textnormal{int}}\ge 0,
\end{equation}
on the one hand, and
\begin{equation}\label{CP therm}
    q\cdot\nx\theta\le 0,
\end{equation}
on the other hand. For reasons that will appear later, we will call \eqref{CP mech} the mechanical/internal Clausius-Planck inequality and \eqref{CP therm} the thermal Clausius-Planck inequality. Now of course, the Clausius-Planck inequalities imply the Clausius-Duhem inequality, but there is no reason for the reverse implication to hold in general. See \cite{HLDAR1} for the identification of some general cases in which the Clausius-Duhem inequality and the Clausius-Planck actually are equivalent.  

\section{The Eulerian Coleman-Noll procedure with an internal variable}\label{section CN}

The dynamics equation \eqref{dynamics eqn} and energy equation \eqref{1st law} are underdetermined partial differential equations, since they are assumed to hold for any (non micropolar) material, in addition to lacking initial and boundary conditions. They primarily need to be complemented with some constitutive information describing what kind of material we are talking about.

As indicated before, we pick a set of thermodynamic variables, namely $h=\nx v$ with values in $\Sl(3)$, $\theta$ with values in $\R_+^*$, $g=\nx \theta$ with values in $\R^3$ and $\xi$ with values in a finite dimensional Euclidean vector space~$V$ with inner product denoted by a $\cdot$ as in $\R^3$  for the time being. Note that we are aiming at liquid crystal models, which are incompressible and basically mechanically fluid, hence the absence of any elastic kind of variable.

The local state hypothesis consists in assuming that all thermodynamic state functions, as well as (most of the) the stress tensor and heat flux vector, depend on  space-time points $(x,t)$ through the values at these same points of the thermodynamic variables via constitutive laws. 

Specifically, using the hat notational device explained earlier, we consider constitutive laws given by functions
\begin{itemize}
\item $\widehat a_m\colon\Sl(3)\times\R_+^*\times\R^3\times V\to\R$ for the Helmholtz free energy specific density,
\item $\widehat s_m\colon\Sl(3)\times\R_+^*\times\R^3\times V\to\R$ for the entropy specific density (combining the two provides a constitutive law $\widehat e_m=\widehat a_m+\theta \widehat s_m$ for the internal energy density),
\item $\widehat q\colon\Sl(3)\times\R_+^*\times\R^3\times V\to\R^3$ for the heat flux vector.
\end{itemize}

There is a slight twist concerning the stress tensor $\sigma$ since it is well known that in the incompressible case, it can only satisfy the local state hypothesis up to \revision{an indeterminate scalar tensor}. Therefore, we also consider a constitutive law for the stress tensor $\widehat \sigma\colon\Sl(3)\times\R_+^*\times\R^3\times V\to\Sym_3$, with the local state hypothesis 
$$
    \sigma(x,t)=\widehat \sigma\bigl((h,\theta,g,\xi)(x,t)\bigr)-p(x,t)I,
$$
\revision{where $p$} is not locally given by a constitutive law. \revision{This indeterminate term is actually the Lagrange multiplier associated with the incompressibility constraint.} Only $\widehat\sigma$ will play a role in the second law considerations below.
 Note that it is not always advisable to require $\widehat\sigma$ to represent the deviatoric part of $\sigma$. For instance, in the case of the Oldroyd B fluid considered in \cite{HLDAR2}, we had $\widehat\sigma(h,\theta,g,\xi)=2\eta_s d+\xi$, and $\tr\xi$ necessarily varied in time and contributed to the dissipation. Therefore, the constitutive part of the stress tensor should definitely not be a priori assumed to be traceless, even though traceless constitutive laws can also be appropriate.

All state functions and fluxes now have a constitutive law.  It is one of the main outcomes of the Coleman-Noll procedure to provide one such law for the internal dissipation \eqref{dissipation interne} or \eqref{dissipation interne 2}.

Since we are dealing with an internal variable theory, we need one more constitutive ingredient to serve as a flow rule for $\xi$, \emph{i.e.}, a function $\widehat k\colon\Sl(3)\times\R_+^*\times\R^3\times V\to V$ as a right-hand side of a differential equation in time for $\xi$,
\begin{equation}\label{edoxi}
    \dot \xi(x,t)=\widehat k \bigl((h,\theta,g,\xi)(x,t)\bigr),
\end{equation}
see \cite{HLDAR1,HLDAR2} and \cite{Maugin} for a panorama of the vast literature on the subject of internal variables. For local existence and uniqueness purposes, we assume that $\widehat k$ is continuous and locally Lipschitz with respect to its last variable uniformly with respect to the others.

The Coleman-Noll procedure \cite{Coleman-Mizel,ColemanNoll} is classically used to derive necessary and sufficient conditions on the constitutive laws ensuring that the resulting model automatically satisfies the second law of thermodynamics, in the best case scenario. Even though its principle is classic, it nonetheless deserves to be performed carefully. Let us go through this procedure in the present context, in a very parallel way to what we did in a Lagrangian setting in \cite{HLDAR1}. 

\begin{proposition}\label{restrictions thermo avec variables internes}
The two laws of thermodynamics and the dynamics equation imply that
 
 i) the specific free energy density $\widehat a_m$ is only a function of $\theta$ and $\xi$,
 
 ii) the specific entropy density $\widehat s_m$ is only a function of $\theta$ and $\xi$ with
\begin{equation}\label{relation entropie}
    \widehat s_m(\theta,\xi)=-\frac{\partial \widehat a_m}{\partial \theta}(\theta,\xi),
\end{equation}

 iii) there is a constitutive law for the internal dissipation given by
\begin{equation}\label{loi dissipation interne}
    \widehat d_{\textnormal{int}}(h,\theta,g,\xi)=\widehat \sigma(h,\theta,g,\xi):h-\rho\frac{\partial\widehat a_m}{\partial \xi}(\theta,\xi)\cdot \widehat k(h,\theta,g,\xi),
\end{equation}
which satisfies  the dissipation inequality
\begin{equation}\label{dissipation ineq}
    \widehat d_{\textnormal{int}}(h,\theta,g,\xi)-\frac{\widehat q(h,\theta,g,\xi)\cdot g}{\theta}\ge 0.
\end{equation}

 Conversely, if the constitutive laws satisfy i), ii) and iii), then the second principle is satisfied for any smooth evolution $(v(x,t),\theta(x,t))$ corresponding to an adapted body force $b(x,t) $ and thermal power source $r(x,t)$.
\end{proposition}

\begin{proof}
Let $\omega\subset\R^3$ be a domain that is filled by part of the material body under consideration at time $t=0$. For any given smooth velocity field $v$ on $\R^3\times\R$, this domain evolves smoothly as $\omega_v(t)$.
The Coleman-Noll procedure consists in testing the Clausius-Duhem inequality with simple fields $v(x,t)$ and $\theta(x,t)$ on $\R^3\times\R$ restricted to $\omega_v(t)$, plus internal variable evolutions with given initial values. Such evolutions can be solutions of the dynamics equation \eqref{dynamics eqn}  and of the energy equation \eqref{1st law}, at least in principle. It is enough to adjust the source terms, namely the applied body force density for the dynamics equation
$$
    b(x,t)=\rho\dot v(x,t)-\div\nolimits_{x}\widehat\sigma\bigl((h,\theta,g,\xi)(x,t)\bigr)+\nx p(x,t)
$$
in $\omega_v(t)$ \revision{for any $p$}, and the heat source 
\begin{multline*}
    r(x,t)=\rho\bigl(\dot{\overbrace{\widehat e_m\bigl((h,\theta,g,\xi)(x,t)\bigr)}}\bigr)-\widehat\sigma\bigl((h,\theta,g,\xi)(x,t)\bigr):d(x,t)\\+{\div\nolimits_x}\bigl(\widehat q\bigl((h,\theta,g,\xi)(x,t)\bigr)\bigr).
\end{multline*}
likewise for the energy equation, as well as surface tractions and heat fluxes on the boundary of $\omega_v(t)$.

We know take the Clausius-Duhem inequality in the form \eqref{2nd law 2} using expression \eqref{dissipation interne 2} for the internal dissipation, use the constitutive laws for the free energy, entropy, stress tensor, heat flux and flow rule for the internal variable~\eqref{edoxi}, and apply the chain rule. For brevity, $\frac{\partial \widehat a_m}{\partial \theta}$ stands for $\frac{\partial \widehat a_m}{\partial \theta}\bigl((h,\theta,g,\xi)(x,t)\bigr)$ and so on. We thus obtain
\begin{equation}\label{ineg fond internes}
    - \rho\Bigl(\frac{\partial \widehat a_m}{\partial \theta}+\widehat s_m\Bigr)\dot \theta+\widehat\sigma:d-\rho\frac{\partial\widehat a_m}{\partial h}:\dot h-\rho\frac{\partial\widehat a_m}{\partial g}\cdot\dot g
    -\rho\frac{\partial\widehat a_m}{\partial \xi}\cdot\widehat k-\frac{\widehat q\cdot g}{\theta}\ge0,
\end{equation}
where the \revision{the Lagrange multiplier $p$} no longer appears due to the incompressibility condition.
 
Let us choose $\bar x\in \omega$ and $\bar h,\bar m\in\Sl(3)$ and set $v(x,t)=(\bar h+t\bar m)(x-\bar x)$. 
We note that $v(\bar x,0)=0$ so that, for any Eulerian quantity $z$,  $\dot z(\bar x,0)=\frac{\partial z}{\partial t}(\bar x,0)$. Since $h(x,t)=\bar h+t\bar m$, it follows that $h(\bar x,0)=\bar h$ and $ \dot h(\bar x,0)=\bar m$.

Next we take $\bar\theta\in \R_+^*$, $\bar\vartheta\in\R$, $\bar g,\bar\ell\in\R^3$, and set $\theta(x,t)=\frac{\bar \theta}2\bigl( \exp\bigl(\frac{2(x-\bar x)\cdot(\bar g+t\bar\ell)}{\vrule height7pt width0pt \bar\theta}\bigr)+\exp\bigl(\frac{2\bar\vartheta t}{\vrule height7pt width0pt \bar\theta}\bigr)\bigr)$. It is easy to check that $\theta(\bar x,0)=\bar\theta$, $\dot\theta(\bar x,0)=\bar\vartheta$, $g(\bar x,0)=\bar g$ and $\dot g(\bar x,0)=\bar \ell$.

Lastly, we take $\bar\xi\in V$ and set $\xi_0(x)=\bar\xi$. 
By the Picard-Lindelöf theorem, there exists $\xi(x,t)$ satisfying the Cauchy problem for equation \eqref{edoxi} with initial datum $\xi_0$ at least locally in time.
 
We assume first $\bar\vartheta=0$ and substitute the above values in \eqref{ineg fond internes} at $(x,t)=(\bar x,0)$. This yields
\begin{multline*}\label{ineq intermediaire}
    \widehat\sigma(\bar h,\bar\theta,\bar g,\bar\xi): \bar h -\rho\Bigl(\frac{\partial\widehat a_m}{\partial \xi}\cdot\hat k\Bigr)(\bar h,\bar\theta,\bar g,\bar\xi)\\
    -\rho\frac{\partial\widehat a_m}{\partial h}(\bar h,\bar\theta,\bar g,\bar\xi):\bar m-\rho\frac{\partial\widehat a_m}{\partial g}(\bar h,\bar\theta,\bar g,\bar\xi)\cdot \bar\ell
    -\frac{\widehat q(\bar h,\bar\theta,\bar g,\bar\xi)\cdot \bar g}{\bar\theta}\ge0.
\end{multline*} 

Since $\bar m\in \Sl(3)$ and $\bar \ell\in\R^3$ are arbitrary, it follows that $\frac{\partial\widehat a_m}{\partial h}=0$ and $\frac{\partial\widehat a_m}{\partial g}=0$. Therefore, $\widehat a_m$ depends neither on $h$, nor on $g$, which is assertion i).

Secondly,  we take the same $v$, $\theta$ and $\xi$ as before with $\bar\vartheta$ arbitrary.
Taking into account assertion i), at point $(\bar x,0)$, inequality \eqref{ineg fond internes} becomes
\begin{multline*}
     -\rho\Bigl(\frac{\partial \widehat a_m}{\partial \theta}(\bar\theta,\bar\xi)+\widehat s_m(\bar h,\bar\theta,\bar g,\bar\xi)\Bigr)\bar\vartheta+\widehat\sigma(\bar h,\bar\theta,\bar g,\bar\xi): \bar h\\ -\rho\frac{\partial\widehat a_m}{\partial \xi}(\bar\theta, \bar\xi)\cdot\hat k(\bar h,\bar\theta,\bar g,\bar\xi)-\frac{\widehat q(\bar h,\bar\theta,\bar g,\bar\xi)\cdot \bar g}{\bar\theta}\ge0.
\end{multline*}
Since $\bar\vartheta$ is arbitrary, it follows that
$$
    \widehat s_m(\bar h,\bar\theta,\bar g,\bar\xi)=-\frac{\partial \widehat a_m}{\partial \theta}(\bar\theta,\bar\xi),
$$
which depends neither on $\bar h$, nor on $\bar g$, and relates the constitutive law for the entropy to that of the Helmholtz free energy, that is to say, assertion ii). 

To conclude, we note that the function $\widehat d_{\textnormal{int}}$ defined in \eqref{loi dissipation interne}, is actually a constitutive law for the internal dissipation $d_{\textnormal{int}}$ defined earlier in \eqref{dissipation interne}, in terms of our chosen thermodynamic variables. 
Indeed, in view of \eqref{relation entropie} and the previous calculations, 
\begin{align*}
    d_{\textnormal{int}}(x,t)&=-\rho(\dot a_m(x,t)+\dot \theta  s_m(x,t))+\sigma(x,t):d(x,t)\\
    &=-\rho\frac{\partial\widehat a_m}{\partial \xi}\bigr((\theta,\xi)(x,t)\bigl)\cdot\widehat k\bigr((h,\theta, g,\xi)(x,t)\bigl)+\widehat\sigma\bigr((h,\theta, g,\xi)(x,t)\bigl):h(x,t)\\
    &=\widehat d_{\textnormal{int}}\bigr((h,\theta, g,\xi)(x,t)\bigl),
\end{align*}
(recalling that $\sigma:d=\widehat\sigma:h$). Assertion iii) is now established.

Conversely, if all these constitutive assumptions hold, then the Clausius-Duhem inequality \eqref{2nd law 2} is  satisfied in all smooth evolutions $v(x,t),\theta(x,t)$. Indeed, it is again enough to express all the state functions using the constitutive laws and adjust the external body force and thermal power source accordingly.
\end{proof}

\begin{remark}
It should be emphasized that, without the Coleman-Noll treatment of the second principle, the internal dissipation does not a priori have a constitutive law in terms of the chosen thermodynamic variables.
\end{remark}

\begin{remark}
As in \cite{HLDAR1}, we notice that if $\widehat\sigma$ and $\widehat k$ do not depend on $g$, and $\widehat q$ does not depend on $h$, then the Clausius-Planck inequalities necessarily hold. We will make this simplifying assumption  in the sequel. In this case, we see that $\widehat d_{\textnormal{int}}$ only contains  mechanical and internal terms, with the temperature as parameter, and $\widehat q\cdot g$ only thermal terms, with the internal variable as parameter, hence the names mechanical/internal and thermal Clausius-Planck inequalities. 
\end{remark}

\begin{remark}
If we now use all the information we obtained on constitutive laws into the dynamics equation and the entropy evolution equation, we get a highly coupled system of partial differential equations in the unknowns $(v,\theta)$,
$$
    \left\{\begin{aligned}
    &\rho \dot v-\div_x\widehat\sigma(h,\theta,\xi)+\nx p=b,\\
    &\div_x v=0,
    \end{aligned}\right.
$$
with 
$\div_x\widehat\sigma(h,\theta,\xi)_i=\mu_{ijmn}(h,\theta,\xi)\frac{\partial^2v_m}{\partial x_j\partial x_n}+\frac{\partial\widehat \sigma_{ij}}{\partial\theta}(h,\theta,\xi)\frac{\partial\theta}{\partial x_j}+\frac{\partial\widehat \sigma_{ij}}{\partial\xi}(h,\theta,\xi)\cdot\frac{\partial\xi}{\partial x_j}$, where the tensor $\mu_{ijmn}(h,\theta,\xi)=\frac{\partial\widehat \sigma_{ij}}{\partial h_{mn}}(h,\theta,\xi)$ is a sort of viscosity tensor, and,
\begin{multline*}
    -\rho\theta\frac{\partial^2\widehat a_m}{\partial\theta^2}(\theta,\xi)\dot\theta+\frac{\partial \widehat q_i}{\partial g_j}(\theta,g,\xi)\frac{\partial^2\theta}{\partial x_j\partial x_i}-\rho\theta\frac{\partial^2\widehat a_m}{\partial\theta\partial\xi}(\theta,\xi)\cdot\widehat k(h,\theta,\xi)\\+\frac{\partial \widehat q_i}{\partial\theta}(\theta,g,\xi) \frac{\partial\theta}{\partial x_i}+\frac{\partial \widehat q_i}{\partial\xi}(\theta,g,\xi)\frac{\partial\xi}{\partial x_i}
    =\widehat d_{\textnormal{int}}(h,\theta,\xi)+r,
\end{multline*}
where the first two terms are evocative of a quasilinear convection-diffusion equation for the temperature, with internal and external heat sources in the right-hand side. 
 
In addition to this, the differential equation for $\xi$ must hold
$$
    \dot \xi(x,t)=\widehat k \bigl((h,\theta,\xi)(x,t)\bigr).
$$
Even when complemented with appropriate initial and boundary conditions, there is nothing in the second law of thermodynamics that can endow this system with some kind of hyperbolicity/parabolicity that could make it mathematically well-posed in general. Additional constitutive hypotheses are required. However, the system is presumably way too general for it to be reasonable to expect much in this direction. It should already be noted that the simplest Newtonian choice  $\widehat\sigma(h,\xi,\theta)=2\mu d$ with dynamic viscosity $\mu>0$, yields the Navier-Stokes equations for the dynamics. 
\end{remark}

\section{MacMillan's internal variable liquid crystal\\ model}\label{section des origines}
We refer to \cite{JMB} for background on general liquid crystal modeling. We are specifically interested in studying and generalizing a dynamic internal variable model introduced by MacMillan in \cite{McM these}. It should be noted that this internal variable model is not the central topic of his work, and that it is therefore not much studied therein. 

We consider here the case when the order parameter is the so-called de~Gennes $Q$-tensor, denoted $S$ in \cite{McM these}, which is a traceless, symmetric valued
tensor. In order to remain consistent with our current notational scheme, we will stray from tradition and call this order tensor $\xi$, with apologies to everyone used to $Q$. As a reminder and for the sake of completeness, this tensor is defined as
\begin{equation}\label{def Q traditionnel}
    \xi (x,t)= \int_{S^2} p\otimes p \, d\mu_{x,t}(p) - \frac 1 3 I,
\end{equation}
where $\mu_{x,t}$ is a probability measure on the unit sphere $S^2$ that describes the distribution of the directions $p$ of the many molecules that are present at a given point  $x$ at time $t$. The tensor $\xi$ thus appears as a second moment of the probability measure $\mu(x,t)$, or autocorrelation matrix of a $S^2$-valued random variable distributed according to $\mu,$ shifted by that of the uniform distribution, in order to quantify the deviation from uniform distribution. It follows from formula \eqref{def Q traditionnel} that $\tr\xi=0$.

Our choice for $V$ is thus now $V=\Sl(3)\cap\Sym_3$, equipped with the Frobenius inner product. Gradients with respect to $\xi$ are meant in this space and relative to this inner product.

In addition to $\xi$ being traceless and symmetric, it is clear from the above definition that $\xi$ should be such that $-\frac13 I\le \xi \le\frac23 I$.  In particular, its eigenvalues $\lambda$ should satisfy the Ericksen inequalities, for instance $-\frac13\le \lambda\le\frac23$, see \cite{Ericksen}. We will not  attempt to capture these boundedness and eigenvalue constraints here.

The internal variable liquid crystal model proposed by MacMillan falls within the scope of our present framework, with some specific choices for the constitutive laws for the stress tensor and for the flow rule. He assumes the stress tensor law to be of the form
\begin{equation}\label{stressMcM}
    \widehat\sigma(h,\theta,\xi)=\widehat\sigma_0(\theta,\xi)+\widehat \sigma_1(\theta,\xi)[d],
\end{equation}
where $\widehat\sigma_0\colon \R_+^*\times V\to \Sym_3$ and $\widehat \sigma_1\colon \R_+^*\times V\to \mathcal{L}(V;\Sym_3)$ are given functions. 
For the flow rule, he takes
\begin{equation}\label{flowruleMcM}
    \widehat k(h,\theta,\xi)=w\xi-\xi w+\widehat k_0(\theta,\xi)+\widehat k_1(\theta,\xi)[d],
\end{equation}
where $\widehat k_0\colon \R_+^*\times V\to V$ and $\widehat k_1\colon \R_+^*\times V\to \mathcal{L}(V;V)$
are also given functions. It is clear that $\widehat k$ is $V$-valued. The evolution of the order tensor is thus governed by the differential equation
\begin{equation}\label{edo ordre}
    \dot \xi(x,t)=w\xi-\xi w+\widehat k_0(\theta,\xi)+\widehat k_1(\theta,\xi)[d].
\end{equation}

In terms of the Helmholtz free energy, his assumption is an unspecified $\widehat a_m(\theta,\xi)$, where we directly remove the a priori dependence on $h$ and $g$ in view of Proposition \ref{restrictions thermo avec variables internes}. 
This includes the bulk energies considered by Ball \cite{JMB},
\begin{equation}\label{bulk JMB}
    \widehat a_m(\theta,\xi)=\widehat\psi(\theta,\tr(\xi^2),\tr(\xi^3)),
\end{equation}
where $\widehat \psi\colon\R_+^*\times\R_+\times \R$ is an arbitrary function,
see also Section \ref{section obj} below. In the sequel, we will always assume that $\widehat\psi$ is as regular as needed.
Among these energies is the de Gennes energy \cite{JMB,PGdG},
\begin{equation}\label{bulk dG}
    \widehat a_m(\theta,\xi)=\alpha(\theta-\theta^*)\tr(\xi^2)-\frac{2b}3\tr(\xi^3)+c\tr\xi^4,
\end{equation}
since  $\tr\xi^4 = \frac12 (\tr(\xi^2))^2$ on $V$, where $\alpha>0$, $b>0$, $c>0$, and $\theta^*>0$ are constants. In this case, we get $\widehat s_m(\theta,\xi)=-\alpha\tr(\xi^2)$ and $\widehat e_m(\theta,\xi)=-\alpha\theta^*\tr(\xi^2)-\frac{2b}3\tr(\xi^3)+c\tr\xi^4$, both being independent of $\theta$.

An important particular case is that of a uniaxial nematic phase in which the order tensor takes the form 
$$
    \xi=s\Bigl(n\otimes n-\frac13I\Bigr),
$$
where $-\frac12\le s\le 1$ is a scalar order parameter and $n$ is a unit vector.

In the uniaxial case, the de Gennes free energy becomes somewhat degenerate,
\begin{equation}\label{bulk dG uniax}
    \widehat a_m(\theta,\xi)=\frac23\alpha(\theta-\theta^*)s^2-\frac{4b}{27}s^3+\frac{2c}{9}s^4,
\end{equation}
as are all bulk energies of the form \eqref{bulk JMB}, which are then only  arbitrary functions of $\theta$ and $s$.
We can still however consider flow rules that take $n$ or $\xi$ into account. 

More interesting free energies for such uniaxial situations, like the Oseen-Frank energy, or more generally Landau-de Gennes energies, involve gradients such as $\nx n$ or $\nx\xi$. We plan to investigate such energies from the internal variable viewpoint in future work.

We also need to add a heat flux vector constitutive law $\widehat q$. As already mentioned, in order to only have to deal with the Clausius-Planck inequalities, we assume that $\widehat q$ does not depend on $h$.

There are further restrictions on the constitutive laws due to objectivity, which we discuss below.

\section{Model objectivity}\label{section obj}
A major requirement for the validity of continuum models is that of objectivity, or frame-indifference \cite{Truesdell-Noll}, \emph{i.e.}, invariance under superimposed translations and rotations. 
Let us thus be given  arbitrary smooth functions  $R$ from $\R$  into  $\SO(3)$ and $z$ from $\R$ into $\R^3$, let $(x^*,t)=(R(t)x+z(t),t)$ for brevity, and attach stars to all translated and rotated quantities. Under such transformations, all scalar-valued functions should be invariant, 
\begin{equation}\label{invariance scalaires}
    \theta^*(x^*,t)=\theta(x,t),\quad a_m^*(x^*,t)=a_m(x,t).
\end{equation}
For vector-valued flux quantities, frame-indifference reads
\begin{equation}\label{invariance vecteurs}
    q^*(x^*,t)=R(t)q(x,t),
\end{equation}
and for tensor-valued quantities such as the Cauchy stress tensor,
\begin{equation}\label{expression AIM}
    \sigma^*(x^*,t)=R(t)\sigma(x,t)R(t)^T.
\end{equation}

The status of any internal variable with respect to frame-indifference depends entirely on what kind of phenomenon the internal variable is supposed to represent. In the case of liquid crystals, in view of the definition \eqref{def Q traditionnel} of the order tensor, we see that it satisfies\cite{JMB,Ericksen}
\begin{equation}\label{invariance xi}
    \xi^*(x^*,t)=R(t)\xi(x,t)R(t)^T.
\end{equation}

Objectivity requirements translate as constraints on the form of the constitutive laws. Those are well-known and listed below. First of all, the constitutive law for the free energy density must be objective, in the sense that for all $\theta\in\R_+^*$, $\xi\in V$, and $R\in\SO(3)$, it must satisfy
\begin{equation}\label{invariance Helmoltz}
    \widehat a_m(\theta,R\xi R^T)=\widehat a_m(\theta,\xi).
\end{equation}
This is the case if and only if $\widehat a_m$ is of the form \eqref{bulk JMB}, see \cite{JMB}, where $\widehat\psi(\theta,\tau_2,\tau_3)$ is a smooth enough function.

For future use, we give the expression of the gradient of $\widehat a_m$ with respect to $\xi$, in terms of this function~$\widehat\psi$, assuming the latter is smooth,
\begin{equation}\label{forme de da/dxi}
    \frac{\partial\widehat a_m}{\partial\xi}(\theta,\xi)=2\frac{\partial\widehat\psi}{\partial\tau_2}(\theta,\tr(\xi^2),\tr(\xi^3))\xi+3\frac{\partial\widehat\psi}{\partial\tau_3}(\theta,\tr(\xi^2),\tr(\xi^3))\dev(\xi^2).
\end{equation}
This follows from the fact that, for any integer $j$,
$$
    \frac{\partial(\tr\xi^j)}{\partial\xi}=j\dev(\xi^{j-1}).
$$
In particular, we always have $ \frac{\partial\widehat a_m}{\partial\xi}(\theta,0)=0$
since $\widehat\psi$ is smooth.

It is shown in \cite{HLDAR1}, in the case of compressible fluids without internal variables, that all fluid, frame-indifferent heat flux constitutive laws must take the form of a nonlinear Fourier law,
$\widehat q(\rho,\theta,g)=-\widehat \kappa(\rho,\theta,\|g\|)g$,
where $\widehat\kappa\colon\R_+^*\times\R_+^*\times\R_+\to\R$ is a given function. This characterization still holds for incompressible fluids with $\rho$ constant, which we can thus remove from the list of arguments of $\widehat q$ and $\widehat\kappa$. If we want to include frame-indifferent heat flux effects due to $\xi$, we can append $\tr(\xi^2)$ and $\tr(\xi^3)$ to this list, thus making $\widehat\kappa$ a real-valued function defined on $\R_+^*\times\R_+\times\R_+\times\R$,
with 
\begin{equation}\label{representation q}
    \widehat q(\theta,g,\xi)=-\widehat \kappa(\theta,\|g\|,\tr(\xi^2),\tr(\xi^3))g.
\end{equation}

The general form of a frame-indifferent, or objective, $\Sym_3$-valued law in two matrix variables, one of which is the velocity gradient $h$, not necessarily incompressible, and the other is a frame indifferent symmetric $\xi$, not necessarily traceless either, is written below in the case of the  constitutive law for the Cauchy stress tensor $\widehat\sigma$, but we will use it again later for the flow rule.

 It is first of all well-known that $\widehat\sigma$ should only depend on $h$ through its stretching tensor $d=\sym(h)$, which is frame-indifferent. Secondly, its general form as a function of $(d,\xi)$ follows from a Rivlin-Ericksen result found in \cite{Rivlin-Ericksen}, slightly corrected in \cite{Rivlin} and then in \cite{Smith}. We follow the latter here,
\begin{multline}\label{stress R-E}
    \widehat\sigma(d,\theta,\xi)=\widehat\alpha_0(d,\theta,\xi)I+\widehat\alpha_1(d,\theta,\xi)d+\widehat\alpha_2(d,\theta,\xi)\xi\\+\widehat\alpha_3(d,\theta,\xi)d^2
    +\widehat\alpha_4(d,\theta,\xi)\xi^2+\widehat\alpha_5(d,\theta,\xi)(\xi d+d\xi)\\
    +\widehat\alpha_6(d,\theta,\xi)(\xi^2 d+d\xi^2)+\widehat\alpha_7(d,\theta,\xi)(\xi d^2+d^2\xi),
\end{multline}
where the scalar-valued functions $\widehat\alpha_i$ are such that
\begin{equation}\label{coeffs RE}
    \widehat\alpha_i(d,\theta,\xi)=\widehat\beta_i\bigl(\theta,\tr d,\tr d^2,\tr d^3,\tr \xi,\tr \xi^2,\tr \xi^3,\tr(\xi d),\tr(\xi d^2),\tr(\xi^2 d),\tr(\xi^2 d^2)\bigr),
\end{equation}
and the functions $\widehat\beta_i\colon \R^*_+\times\R^{10}\to\R$ are arbitrary (in our case, we have $\tr d=\tr\xi=0$). The constitutive law~\eqref{stressMcM} proposed by MacMillan for the stress tensor is thus objective as soon as
\begin{equation}\label{isotropie 0}
    \widehat\sigma_0(\theta,\xi)=\sum_{j=0}^2\widehat\gamma_j(\theta,\tr(\xi^2),\tr(\xi^3))\xi^j,
\end{equation}
by taking $d=0$, and
\begin{multline}\label{isotropie 1}
    \widehat\sigma_1(\theta,\xi)[d]=\tr(\xi d)\sum_{j=0}^2\widehat\gamma'_j(\theta,\tr(\xi^2),\tr(\xi^3))\xi^j
    +\tr(\xi^2 d)\sum_{j=0}^2\widehat\gamma''_j(\theta,\tr(\xi^2),\tr(\xi^3))\xi^j\\
    +\sum_{j=0}^2\widehat\gamma'''_j(\theta,\tr(\xi^2),\tr(\xi^3))(\xi^jd+d\xi^j),
\end{multline}
for some arbitrary scalar-valued functions $\widehat\gamma_j$, $\widehat\gamma'_j$, $\widehat\gamma''_j$ and $\widehat\gamma'''_j$. Formula \eqref{isotropie 1}, which is clearly sufficient, can be shown to hold necessarily under mild boundedness or regularity assumptions on some of the functions $\widehat\beta_i$. Proving its necessity in all generality would probably require revisiting the original Rivlin-Ericksen argument.

Finally, the objectivity of $\xi$ expressed by \eqref{invariance xi} requires the introduction of objective derivatives in the flow rule \eqref{edoxi}.
Let us recall a few well-known facts about objective derivatives.

An objective derivative $\bigpenta$ is a differential operator that is of first order in time and respects frame-indifference,
\begin{equation*}
    \overset{\pentagon^*}{\xi^{\smash{*}}}(x^*,t)=R(t)\overset{\penta}{\xi}(x,t)R(t)^T,
\end{equation*}
for all functions $\xi$ and $R$ with values in $\Sym_3$ and $\SO(3)$  respectively, where $\xi^*$ and $\xi$ are related via \eqref{invariance xi}.

There is a general description of objective derivatives of the form
\def\Ob{\mathrm{Ob}}
\begin{equation}\label{forme a priori}
    \overset{\penta}{\xi}=\dot \xi+\Ob(h,\theta,\xi),
\end{equation}
with $\Ob\colon\M_3\times\R_+^*\times \Sym_3\to\Sym_3$ that is given in \cite{Gurtinandco} (albeit without the temperature, which is an objective scalar). 

\begin{proposition}\label{description des derivees objectives}
An operator of the form \eqref{forme a priori} is an objective derivative if and only if
\begin{equation}\label{forme a posteriori}
    \Ob(h,\theta,\xi)=\xi w- w\xi -\widehat k_s(d,\theta,\xi).
\end{equation}
where 
$\widehat k_s\colon \Sym_3\times\R_+^*\times \Sym_3\to\Sym_3$ is an objective function.
\end{proposition}

In MacMillan's case, the differential equation \eqref{edo ordre} can be rewritten as
$$
    \overset{\squareabove}{\xi}=\widehat k_0(\theta,\xi)+\widehat k_1(\theta,\xi)[d],
$$
where 
$$
    \overset{\squareabove}{\xi}=\dot\xi+\xi w-w\xi
$$
is the Zaremba-Jaumann \cite{Zaremba} or corotational derivative of $\xi$ (with $\widehat k_s=0$), see also \cite{MauginDrouot} in a related context. Now, this flow rule can also be rewritten as 
\begin{equation}\label{flowruleMcM reecrite 1}
    \overset{\penta}{\xi}=0,
\end{equation}
where
\begin{equation}\label{flowruleMcM reecrite 2}
    \overset{\penta}{\xi}=\dot \xi+ \xi w-w\xi -\widehat k_0(\theta,\xi)-\widehat k_1(\theta,\xi)[d],
\end{equation}
where $\widehat k_0$ and $\widehat k_1$ are objective of the general form \eqref{isotropie 0} and \eqref{isotropie 1} respectively. In view of Proposition~\ref{description des derivees objectives}, this is an objective derivative, with $\widehat k_s(d,\theta,\xi)=\widehat k_0(\theta,\xi)+\widehat k_1(\theta,\xi)[d]$. Therefore, this flow rule generates a frame-indifferent internal variable $\xi$.
There is thus nothing special in the appearance of the Zaremba-Jaumann derivative in MacMillan's flow rule, nor in \cite{MauginDrouot,Hand}. 

We can actually generalize it by adopting \eqref{flowruleMcM reecrite 1} as a general frame-indifferent flow rule, using any objective derivative $\bigpenta$ given by \eqref{forme a posteriori}. This just means that we set 
$$
    \widehat k(h,\theta,\xi)=w\xi-\xi w+\widehat k_s(d,\theta,\xi),
$$
in the initial formulation of the flow rule.

Clearly, no objective derivative is a priori better than any other in this context, and the choice of a specific $\bigpenta$ is purely a modeling choice at this stage.

We remark that the use of an objective derivative in the flow rule makes the internal dissipation constitutive law \eqref{loi dissipation interne} frame-indifferent, as it should be. We actually note a slight simplification in this  law, also noted in \cite{McM these}, which is due to objectivity.
\begin{proposition}\label{une premiere simplification}
We have 
\begin{equation}\label{plus concis}
     \widehat d_{\textnormal{int}}(h,\theta,\xi)=\widehat \sigma(d,\theta,\xi):d-\rho\frac{\partial\widehat a_m}{\partial \xi}(\theta,\xi): \widehat k_s(d,\theta,\xi).
\end{equation}
\end{proposition}  
\begin{proof}
It follows from \eqref{forme de da/dxi}, 
that
$\frac{\partial\widehat a_m}{\partial\xi}(\theta,\xi)$ commutes with $\xi$. Therefore, for all $w$,
\begin{align*}
    \frac{\partial\widehat a_m}{\partial\xi}(\theta,\xi):\bigl(w\xi-\xi w)&=
    \tr\Bigl(\frac{\partial\widehat a_m}{\partial\xi}(\theta,\xi)w\xi\Bigr)-\tr\Bigl(\frac{\partial\widehat a_m}{\partial\xi}(\theta,\xi)\xi w\Bigr)\\
    &=
    \tr\Bigl(\frac{\partial\widehat a_m}{\partial\xi}(\theta,\xi)w\xi\Bigr)-\tr\Bigl(\xi\frac{\partial\widehat a_m}{\partial\xi}(\theta,\xi) w\Bigr)=0,
\end{align*}
since $\tr(AB)=\tr(BA)$ for all $A,B$, from which the result follows.
\end{proof}

In other words, the Zaremba-Jaumann part of the objective derivative flow rule does not play any role in the dissipation.

The above discussion does not take into account the fact that $\xi$ must be traceless.
In the next section, we discuss the objective derivatives that preserve the zero trace condition in all generality.
   
\section{Traceless objective derivatives}\label{section traceless}
We  say that an objective derivative $\bigpenta$ preserves the zero trace condition, or is a traceless objective derivative, if given an initial value $\xi_0$ with values in $V$, the corresponding solution $\xi$ of the differential equation $\overset{\penta}{\xi}=0$ still has values in $V$ for any traceless velocity gradient $h$. In this respect, we have the following proposition.
\begin{proposition}\label{lequelles sont traceless}
An objective derivative of the form \eqref{forme a priori}, \eqref{forme a posteriori} conserves the zero trace condition if and only if 
$\tr\bigl(\widehat k_s(d,\theta,\xi)\bigr)=0$ for all $(d,\xi)\in V^2$, $\theta\in\R_+^*$.
\end{proposition}
\begin{proof}
Let $\bigpenta$ be an objective derivative and assume that, if we are given an initial condition $\xi_0\colon\omega\to V$, the flow rule \eqref{flowruleMcM reecrite 1} produces a traceless $\xi(x,t)$ for any evolution $(v(x,t),\theta(x,t))$.  Obviously, we then have $\tr(\dot\xi)=0$ and $\tr(\xi w-w\xi)=\tr(\xi w)-\tr(w\xi)=0$. Consequently,
$$
    0=\tr\Bigl(\overset{\penta}{\xi}\Bigr)=-\tr\bigl(\widehat k_s(d,\theta,\xi)\bigr),
$$
for all $(x,t)$. Let us take an arbitrary $\bar\xi\in V$ and define $\xi_0(x)=\bar\xi$. Choosing $\bar x\in\omega$, we can also adjust $d(\bar x,0)=\bar d$ for any $\bar d\in V$ and $\theta(\bar x,0)=\bar\theta$ for any $\bar\theta_0\in \R_+^*$, as we did before in the proof of Proposition~\ref{restrictions thermo avec variables internes}, therefore $\tr\bigl(\widehat k_s(\bar d,\bar\theta,\bar\xi)\bigr)=0$ for all $(\bar d,\bar\theta,\bar\xi)\in V\times\R_+^*\times V$.

Conversely, assuming $\tr\bigl(\widehat k_s(d,\theta,\xi)\bigr)=0$ and writing the flow rule as $\dot \xi=k(x,t,\xi)$ with $k(x,t,\xi)=w(x,t)\xi-\xi w(x,t)+\widehat k_s(d(x,t),\theta(x,t),\xi)$, we see that we are dealing with an ordinary differential equation, the right-hand side of which has values in the vector space $V$. Assuming an initial value in $V$, then $\xi(x,t)\in V$ for all $(x,t)$.

\end{proof}

For the flow rule \eqref{flowruleMcM reecrite 1}--\eqref{flowruleMcM reecrite 2}, MacMillan assumes $\widehat k_0$ and $\widehat k_1$ to map into $V$. According to Proposition~\ref{lequelles sont traceless}, an initial $\xi_0$ in $V$ will thus evolve in $V$ as expected.

To represent all traceless objective derivatives of the form \eqref{forme a priori}, it is obviously enough to take any objective function %$\Ob_s$ 
$\widehat k_s$ as in \eqref{stress R-E}, and project it on $V$ by replacing it with $\widehat k_s-\frac{\tr(\widehat k_s)}{3}I=\dev(\widehat k_s)$, which is also objective. 

Let us see what this gives for a few common objective derivatives.
\begin{itemize}
\item The Zaremba-Jaumann derivative with $\widehat k_s=0$ is already traceless without modification. 
\item We have a traceless Oldroyd A derivative
$\overset{\vartriangle_0}{\xi}=\dot \xi+h^T\!\xi+\xi h-\frac23\tr(\xi d)I$ 
and a traceless Oldroyd B derivative
$\overset{\triangledown_{\!0}}{\xi}=\dot \xi-h\xi-\xi h^T+\frac23\tr(\xi d)I$. The original A and B cases, $\overset{\vartriangle}{\xi}$ and $\overset{\triangledown}{\xi}$, neither one of which is traceless, were introduced in \cite{Oldroyd}.
\end{itemize}

More importantly perhaps, following are representation formulas for the $\widehat k_0$ and $\widehat k_1$ functions appearing in MacMillan's flow rule. These are just obtained by projecting \eqref{isotropie 0} and \eqref{isotropie 1} on $V$ as explained above, which gives
\begin{equation}\label{isotropie 2}
    \widehat k_0(\theta,\xi)=\sum_{j=1}^2\widehat\delta_j\bigl(\theta,\tr(\xi^2),\tr(\xi^3)\bigr)\dev(\xi^j),
\end{equation}
and
\begin{multline}\label{isotropie 3}
    \widehat k_1(\theta,\xi)[d]=\tr(\xi d)\sum_{j=1}^2\widehat\delta'_j\bigl(\theta,\tr(\xi^2),\tr(\xi^3)\bigr)\dev(\xi^j)
    +\tr(\xi^2 d)\sum_{j=1}^2\widehat\delta''_j\bigl(\theta,\tr(\xi^2),\tr(\xi^3)\bigr)\dev(\xi^j)\\
    +\sum_{j=0}^2\widehat\delta'''_j\bigl(\theta,\tr(\xi^2),\tr(\xi^3)\bigr)\dev(\xi^jd+d\xi^j),
\end{multline}
keeping in mind that $\dev(I)=0$, $\dev(\xi)=\xi$ and $\dev(d)=d$, for some arbitrary scalar-valued functions $\widehat\delta_j$, $\widehat\delta'_j$, $\widehat\delta''_j$ and $\widehat\delta'''_j$, thus yielding the aforementioned traceless objective derivatives. 

There are related considerations in \cite{Hand}. 

\section{Compatibility of MacMillan's model with the second law of thermodynamics}\label{section McM compatible}
In \cite{McM these}, MacMillan presents a rather cursory thermodynamical analysis of his internal variable model. He gives two, not entirely explicit, 
 necessary conditions for the second law to hold. In this section, we go beyond this both in terms of generality and of sufficiency. We also work out a few simple examples. 

We are in a case in which the Clausius-Planck inequalities are relevant. The thermal Clausius-Planck inequality is fairly easy to deal with. Indeed, it is clearly necessary and sufficient that $\widehat \kappa\ge 0$, where $\widehat\kappa$ is the function appearing in the nonlinear Fourier law \eqref{representation q}.

Let us then turn to the mechanical/internal Clausius-Planck inequality,
\begin{equation}\label{CP mech const}
    \widehat d_{\textnormal{int}}(h,\theta,\xi)\ge0,
\end{equation}
for all $(h,\theta,\xi)\in\M_3\times\R_+^*\times V$. For brevity, in computations dealing with this inequality, we will not write the argument $\theta$, which only plays the role of a parameter.

It is very well-known, and quite clear from expression \eqref{plus concis}, that choosing a free energy and a flow rule  independently of each other cannot produce a model that satisfies the second law.
Since the literature is already very large concerning liquid crystal bulk free energies, we thus opt to adapt the flow rules to the already accepted free energies, which are supposed to capture part of the physics behind liquid crystals. 

 For future reference, we will use the notation $\tau_2$ as a placeholder for $\tr(\xi^2)$ and $\tau_3$ for $\tr(\xi^3)$. We note that the variables $\tau_2$ and $\tau_3$ are constrained as follows.
\begin{lemma}\label{domaine}
Let $\mathcal{D}=\bigl\{(\tr(\xi^2),\tr(\xi^3));\xi\in V\bigr\}$.
We have  $\mathcal{D}=\bigl\{(\tau_2,\tau_3)\in \R_+\times\R;|\tau_3|\le\frac{1}{\sqrt6}\tau_2^{3/2}\bigr\}$.
\end{lemma}

\begin{proof}
Let $P(X)=X^3+\tr(\cof\xi)X-\det\xi$ be the characteristic polynomial of $\xi$. We note that, since $\xi\in V$, $\tr(\cof\xi)=-\frac12\tau_2$ and $\det\xi=\frac13\tau_3$. Since $\xi$ is symmetric, all the roots of its characteristic polynomial are real, hence its discriminant is nonnegative, \emph{i.e.},
$$
    \frac12\tau_2^3-3\tau_3^2\ge 0.
$$
Since $\tau_2\ge0$, it follows that $(\tau_2,\tau_3)\in\mathcal{D}$. Conversely, given $(\tau_2,\tau_3)\in\mathcal{D}$, then $P$ has real roots $(\lambda_1,\lambda_2,\lambda_3)\in\R^3$, and the matrix $\xi=\diag(\lambda_1,\lambda_2,\lambda_3)$ belongs to $V$ and is such that $\tau_2=\tr(\xi^2)$ and $\tau_3=\tr(\xi^3)$. 
\end{proof}

Let us consider MacMillan's 
constitutive assumptions \eqref{stressMcM} and \eqref{flowruleMcM}. In this case, and using Proposition~\ref{une premiere simplification}, inequality \eqref{CP mech const} becomes
\begin{equation}\label{CP McM}
    \bigl(\widehat\sigma_0(\theta,\xi)+\widehat\sigma_1(\theta,\xi)[d]\bigr):d-\rho\frac{\partial\widehat a_m}{\partial\xi}(\theta,\xi):\bigl(\widehat k_0(\theta,\xi)+\widehat k_1(\theta,\xi)[d]\bigr)\ge0.
\end{equation}

\begin{proposition}\label{les conditions pour McM}
A set of necessary and sufficient conditions for MacMillan's internal variable model to satisfy the second law of thermodynamics is 
\begin{gather}
    \frac{\partial\widehat a_m}{\partial\xi}(\theta,\xi):\widehat k_0(\theta,\xi)\le0,
    \label{condition 1}\\
    \bigl(\widehat\sigma_1(\theta,\xi)[\bar d\,]\bigr):\bar d\ge0,
    \label{condition 2}\\
    \Bigl(\widehat\sigma_0(\theta,\xi):\bar d-\rho\frac{\partial\widehat a_m}{\partial\xi}(\theta,\xi):\bigl(\widehat k_1(\theta,\xi)[\bar d\,]\bigr)\Bigr)^2\qquad\qquad\nonumber\\\qquad\qquad+4\rho\Bigl(\frac{\partial\widehat a_m}{\partial\xi}(\theta,\xi):\widehat k_0(\theta,\xi)\Bigr)\bigl(\bigl(\widehat \sigma_1(\theta,\xi)[\bar d\,]\bigr):\bar d\bigr)\le0.
    \label{condition 3}
 \end{gather}
 and 
\begin{equation}\label{condition thermique}
    \widehat \kappa(\theta,n,\tau_2,\tau_3)\ge 0,
\end{equation}
for all $\bar d\in V$, $\|\bar d\|=1$, $\theta\in \R_+^*$, $\xi\in V$,   $n\in\R_+$,  $(\tau_2,\tau_3)\in\mathcal{D}$.
\end{proposition}
\begin{proof}
For the mechanical/internal Clausius-Planck inequality, we take $\bar d\in V$ with $\|\bar d\|=1$, and set $d=\lambda\bar d$ with $\lambda\in\R$. Then, inequality
\eqref{CP McM} becomes (omitting $\theta$),
\begin{equation}\label{CP McM bis}
    \lambda\bigl(\widehat\sigma_0(\xi)+\lambda\widehat\sigma_1(\xi)[\bar d\,]\bigr):\bar d-\rho\frac{\partial\widehat a_m}{\partial\xi}(\xi):\bigl(\widehat k_0(\xi)+\lambda\widehat k_1(\xi)[\bar d\,]\bigr)\ge0,
\end{equation}
the left-hand side of which is a polynomial of degree at most $2$ in $\lambda$. Conditions \eqref{condition 1}, \eqref{condition 2} and \eqref{condition 3} are necessary and sufficient for this polynomial to only take nonnegative values. 

We have already mentioned that, more generally, the thermal Clausius-Planck inequality is equivalent to \eqref{condition thermique}.
\end{proof}

\begin{remark}
Conditions \eqref{condition 1} and \eqref{condition 2} are the necessary conditions found by MacMillan in \cite{McM these}. Condition~\eqref{condition 3} seems to be new.
\end{remark}

We are now going to examine the conditions of Proposition \ref{les conditions pour McM} one after another,  considering the free energy as given as explained earlier.

Let us express condition \eqref{condition 1} in terms of the representation formulas \eqref{forme de da/dxi} and \eqref{isotropie 2}. We will use the following auxiliary functions stemming from our given $\widehat\psi$:
\begin{align}
    \widehat\psi_1^*(\theta,\tau_2,\tau_3)=2\tau_3\frac{\partial\widehat\psi}{\partial\tau_2}(\theta,\tau_2,\tau_3)+\frac12 \tau_2^2 \frac{\partial\widehat\psi}{\partial\tau_3}(\theta,\tau_2,\tau_3),\label{psi*un}\\
    \widehat\psi_2^*(\theta,\tau_2,\tau_3)=2\tau_2\frac{\partial\widehat\psi}{\partial\tau_2}(\theta,\tau_2,\tau_3) + 3 \tau_3 \frac{\partial\widehat\psi}{\partial\tau_3}(\theta,\tau_2,\tau_3).\label{psi*deux}
\end{align}

For brevity, we mostly omit the arguments $(\theta, \tr(\xi^2), \tr(\xi^3))$ or $(\theta, \tau_2,\tau_3)$ depending on context.
\begin{proposition}\label{premiere condition k0}
A function $\widehat k_0$ satisfies condition  \eqref{condition 1} if and only if its coefficients $\widehat\delta_1,\widehat\delta_2$ are such that
\begin{equation}\label{underwhelming}
    \widehat\psi_2^*\widehat\delta_1+\widehat\psi^*_1\widehat\delta_2\le0,
\end{equation}
for all $(\theta,\tau_2,\tau_3)$ in $\R_+^*\times \mathcal{D}$.
\end{proposition}

\begin{proof}
Straightforward expansion, together with the identities $\|\xi\|^2=\tr(\xi^2)$, $\xi:\dev(\xi^2)=\tr(\xi^3)$ and $\|\dev(\xi^2)\|^2=\frac16\bigl(\tr(\xi^2)\bigr)^2$. 
\end{proof}

In spite of \eqref{underwhelming} being necessary and sufficient, it is still useful to have more readily usable conditions.
In this respect, we have the following result.

\begin{proposition}\label{une forme explicite non mais !}
Let $\widehat\psi$ be given. A function $\widehat k_0$ satisfies \eqref{condition 1} if and only if it has the form
\begin{equation}\label{forme de k0}
    \widehat k_0(\theta,\xi)=-\widehat\eta_0(\theta,\xi)\frac{\partial\widehat a_m}{\partial\xi}(\theta,\xi)+\widehat k_0^*(\theta,\xi),
\end{equation}
when $\frac{\partial\widehat a_m}{\partial\xi}(\theta,\xi)\neq0$, with 
$\widehat\eta_0$ nonnegative and objective, and
\begin{equation}\label{forme de z0}
     \widehat k_0^*(\theta,\xi)=\widehat\omega(\theta,\xi)\Bigl(\widehat\psi^*_1(\theta,\tr(\xi^2),\tr(\xi^3))\,\xi
     -\widehat\psi^*_2(\theta,\tr(\xi^2),\tr(\xi^3))\dev(\xi^2)\Bigr)\Bigr),
\end{equation}
with $\widehat\omega$ scalar-valued and objective,
when $\xi$ is not uniaxial, and
$\widehat k_0^*(\theta,\xi)=0$ when $\xi$ is uniaxial. In particular, $\frac{\partial\widehat a_m}{\partial\xi}(\theta,\xi)$ and $\widehat k_0^*(\theta,\xi)$ are orthogonal.
\end{proposition}

\begin{proof}
When $\frac{\partial\widehat a_m}{\partial\xi}(\theta,\xi)=0$, inequality \eqref{condition 1} is satisfied irrespective of $\widehat k_0$.
Otherwise, we can always write $\widehat k_0(\theta,\xi)=-\widehat\eta_0(\theta,\xi)\frac{\partial\widehat a_m}{\partial\xi}(\theta,\xi)+\widehat k_0^*(\theta,\xi)$, with $\widehat k_0^*(\theta,\xi)$ orthogonal to $\frac{\partial\widehat a_m}{\partial\xi}(\theta,\xi)$ and $\widehat\eta_0$ some scalar objective function.  Since $\frac{\partial\widehat a_m}{\partial\xi}(\theta,\xi):\widehat k_0(\theta,\xi)=-\widehat\eta_0(\theta,\xi)\bigl\|\frac{\partial\widehat a_m}{\partial\xi}(\theta,\xi)\bigr\|^2$, condition \eqref{condition 1} just says that $\widehat \eta_0\ge0$. 

Note that by formulas \eqref{forme de da/dxi} and \eqref{isotropie 2}, both $\frac{\partial\widehat a_m}{\partial\xi}(\theta,\xi)$ and $\widehat k_0(\theta,\xi)$ belong to $\vect(\xi,\dev(\xi^2))$.  Therefore $\widehat k_0^*(\theta,\xi)$ also belongs to $\vect(\xi,\dev(\xi^2))$, a space  which is of dimension at most two. 
It is not hard 
to see that $(\xi,\dev(\xi^2))$ are linearly dependent if and only if $\xi$ is uniaxial. Therefore, it follows that $\widehat k_0^*=0$ when $\xi$ is uniaxial. When $(\xi,\dev(\xi^2))$ are linearly independent, that is to say $\xi$ is not uniaxial, then formula \eqref{forme de z0} is just a standard parametric representation of the vector line  $\frac{\partial\widehat a_m}{\partial\xi}(\theta,\xi)^\bot$ in $\vect(\xi,\dev(\xi^2))$.
\end{proof} 

Thus, the general form of $\widehat k_0$ depends on two arbitrary, scalar-valued and objective functions $\widehat \eta_0$ and $\widehat \omega$, with $\widehat\eta_0\ge0$.
In particular, a convenient sufficient condition is to just take $\widehat\omega=0$ and
\begin{equation}\label{convenient}
    \widehat k_0(\theta,\xi)=-\widehat\eta_0(\theta,\xi)\frac{\partial\widehat a_m}{\partial\xi}(\theta,\xi)\text{ with }\widehat\eta_0(\theta,\xi)\ge0 \text{ and objective.}
\end{equation}

Naturally, formula \eqref{forme de z0} is also a complicated way of writing $\widehat k_0^*=0$ when $\xi$ is uniaxial.

We next note that conditions \eqref{condition 2} and \eqref{condition 3} say that two quadratic forms on $V$ must be nonnegative for the first one, and nonpositive for the second one, for all $(\theta,\xi)$. 
We will make repeated use of the following lemma which gives two different sets of sufficient conditions for a quadratic expression to remain nonnegative.
\begin{lemma}\label{un poil moins trivial}
Let $\alpha,\beta,\gamma\in\R$ and let $\xi$ in $V$. For all $d\in V$, let $P_\xi(d)=\alpha\|d\|^2+\beta\xi d : d+\gamma\|\xi d\|^2$. 
 
 i) if $\alpha\ge 0$, $\gamma\ge 0$, and $\beta^2\le 4\alpha\gamma$, 
 
 \noindent or
 
 ii) if $\|\beta\xi+\gamma\xi^2\|=\bigl(\tau_2 \beta^2+\frac{\tau_2^2}{2}\gamma^2+2\tau_3{\,}\beta\gamma\bigr)^{\frac12}\le \alpha$, 
 
\noindent then $P_\xi(d)\ge 0$ for all $d\in V$.
 \end{lemma}

\begin{proof}
Conditions i) are well-known. We write the proof for completeness. When $\alpha=0$, these conditions imply $\beta=0$, $\gamma\ge 0$ and the result follows. 
If $\alpha>0$, it suffices to write 
$$\alpha P_\xi(d)=\Bigl\|\alpha d+\frac \beta2 \xi d\Bigr\|^2+\Bigl(\alpha \gamma-\frac{\beta^2}4\Bigr)\|\xi d\|^2,$$
which is then the sum of two nonnegative terms. 

Now for condition ii), we can also write $P_\xi(d)$ as 
$$P_\xi(d)=\alpha\|d\|^2+\bigl(\beta{\,}\xi d+\gamma{\,}\xi^2d\bigr):d=\alpha\|d\|^2+\bigl(\beta{\,}\xi+\gamma{\,}\xi^2\bigr)d:d.$$
Therefore,
$$
    P_\xi(d)\ge \alpha\|d\|^2-\bigl\|\bigr(\beta{\,}\xi+\gamma{\,}\xi^2\bigr)d\|\|d\|\ge\bigl(\alpha-\|\beta\xi+\gamma\xi^2\|\bigr)\|d\|^2,
$$
since the Frobenius norm is submultiplicative, which shows that $P_\xi(d)$ remains nonnegative for all $d$ when condition ii) is satisfied.
\end{proof}

\begin{remark}  Conditions i) and ii) are not comparable. For instance, conditions i) are satisfied when $\alpha\ge 0, \gamma\ge 0, \beta=0$ while conditions ii) may be not satisfied for some $\xi$. Similarly, condition ii) may be satisfied with $\beta\neq 0$, $\gamma=0$, while conditions i) cannot hold.

In the sequel, we will apply Lemma \ref{un poil moins trivial} in cases when $\alpha$, $\beta$ and $\gamma$ are also functions of $\xi$.
\end{remark}

Let us now give explicit sufficient conditions, expressed in terms of the $\widehat \gamma$ functions that appear in the representation formula \eqref{isotropie 1} for $\widehat\sigma_1$, for condition \eqref{condition 2} to be satisfied.
\begin{proposition}\label{sufficient form 1}
If we assume that, for all arguments $(\theta, \tau_2,\tau_3)$,  $\widehat\gamma_1'\ge0,\, \widehat\gamma_2''\ge0$, $(\widehat\gamma_2'+\widehat\gamma_1'')^2\le4\widehat\gamma'_1\widehat\gamma''_2$, and that,
\begin{equation} \label {le premier premier choix}
    \widehat\gamma_0'''\ge0,\,\widehat\gamma_2'''\ge0,\,(\widehat\gamma_1''')^2\le4\widehat\gamma'''_0\widehat\gamma'''_2, 
\end{equation}
or 
\begin{equation}\label{le premier deuxième choix}
    \bigl(\tau_2 (\widehat\gamma_1''')^2+\frac{\tau_2^2}{2}(\widehat\gamma'''_2)^2+2\tau_3{\,}\widehat\gamma_1'''\widehat\gamma'''_2\bigr)^{1/2}\le \widehat\gamma_0''',
\end{equation}
then condition \eqref{condition 2} is satisfied.
\end{proposition}

\begin{proof}
Let us rewrite \eqref{isotropie 1} as follows (still omitting the arguments $(\theta, \tr(\xi^2), \tr(\xi^3))$ for brevity):
$$
    \widehat\sigma_1(\xi)[d]=\tr(\xi d)\widehat P_1(\xi)+\tr(\xi^2 d)\widehat P_2(\xi)+ 2{\,}\widehat\gamma'''_0 d+\widehat\gamma'''_1(\xi d+d \xi) +\widehat\gamma'''_2(\xi^2 d+d \xi^2),
$$
with 
$$
    \widehat P_1(\xi)=\widehat\gamma_0'I+\widehat\gamma_1'\xi+\widehat\gamma_2'\xi^2\text{ and }\widehat P_2(\xi)=\widehat\gamma_0''I+\widehat\gamma_1''\xi+\widehat\gamma_2''\xi^2.
$$
Let us set
\begin{equation*}
    \widehat\sigma_1(\xi)[d]=\widehat\sigma'_1(\xi)[d]+\widehat\sigma'''_1(\xi)[d],
\end{equation*}
with
\begin{equation}
    \widehat\sigma'_1(\xi)[d]:d=\tr(\xi d)\widehat P_1(\xi):d+\tr(\xi^2 d)\widehat P_2(\xi):d,
\end{equation}
and
\begin{equation}\label{sigma tierce}
    \widehat\sigma_1'''(\xi)[d]:d=2\bigl(\widehat\gamma'''_0 d:d +\widehat\gamma_1'''\xi d:d+\widehat\gamma'''_2 \xi d: \xi d\bigr).
\end{equation}
since $\xi^2d:d=d\xi^2:d=\xi d:\xi d$.

Our goal now is to give conditions on the functions $\widehat\gamma$ that will make both of these terms separately always nonnegative.
First of all, we have 
\begin{equation*}
    \widehat\sigma'_1(\xi)[d]:d=\widehat\gamma_1'\bigl(\tr(\xi d)\bigl)^2+(\widehat\gamma_2'+\widehat\gamma_1'')\tr(\xi d)\tr(\xi^2 d)
    +\widehat\gamma_2''\bigl(\tr(\xi^2 d)\bigl)^2.
\end{equation*}
The conditions $\widehat\gamma_1'\ge0$, $\widehat\gamma_2''\ge0$ and $(\widehat\gamma_2'+\widehat\gamma_1'')^2-4\widehat\gamma'_1\widehat\gamma''_2\le 0$ make it nonnegative.  
For the second term, we use Lemma \ref{un poil moins trivial} with expression \eqref{sigma tierce}. Case i) corresponds to \eqref{le premier premier choix} and case ii) to \eqref{le premier deuxième choix}.
\end{proof}

\begin{remark}
 The functions $\widehat\gamma'_0$ and $\widehat\gamma''_0$ obviously play no role in the internal dissipation, and are thus not subjected to any constraint related to the second law.
\end{remark}

The quadratic form in inequality \eqref{condition 3} 
is perhaps
less accessible than the first one in \eqref{condition 2} in all generality. 

First of all, the term $\widehat\sigma_0$ corresponds to a residual Cauchy stress at rest $d=0$, that incorporates temperature and order tensor effects. Its deviatoric part belongs to $\vect(\xi,\dev(\xi^2))$, so that we can write 
\begin{equation}\label{dev sigma0}
    \dev(\widehat\sigma_0(\theta,\xi))=2\rho\widehat\lambda_0(\theta,\xi)\frac{\partial\widehat a_m}{\partial \xi}(\theta,\xi)+\widehat\sigma_0^*(\theta,\xi),
\end{equation}
for some scalar-valued, objective function $\widehat\lambda_0$,  with $\widehat\sigma_0^*$ as in Proposition
\ref{une forme explicite non mais !}, that is to say that $\widehat\sigma_0^*(\theta,\xi)$ is orthogonal to $\frac{\partial\widehat a_m}{\partial \xi}(\theta,\xi)$ in $\vect(\xi,\dev(\xi^2))$.

We consider the simpler case when $\widehat\sigma_1(\theta,\xi)=2\widehat\gamma_0'''(\theta, \tr(\xi^2), \tr(\xi^3))\textnormal{id}$, with $\widehat\gamma_0'''\ge0$ as per Proposition~\ref{sufficient form 1}. 
This produces a Newtonian stress term with viscosity dependent on temperature and order tensor, and we switch to the usual notation for the dynamic viscosity coefficient, \emph{i.e.}, 
$$
    \widehat\sigma_1(\theta,\xi)[d]=2\widehat\mu_1(\theta,\xi)d,
$$ 
where $\widehat\mu_1(\theta,\xi)=\widehat\gamma_0'''(\theta, \tr(\xi^2), \tr(\xi^3))$.

Next, following \eqref{convenient}, we make the assumption $\widehat k_0(\theta,\xi)=-\widehat\eta_0(\theta,\xi)\frac{\partial\widehat a_m}{\partial\xi}(\theta,\xi)$ with $\widehat\eta_0\ge0$. This is without real loss of generality in terms of internal dissipation, due to Proposition \ref{une forme explicite non mais !}.

Finally, still in a spirit of simplification, we assume that 
$$
    \widehat k_1(\theta,\xi)=2\widehat\eta_1(\theta, \xi)\textnormal{id},
$$
with yet another scalar-valued, objective function $\widehat\eta_1$.

Note that the scalar-valued functions $\widehat\lambda_0$ and $\widehat\mu_1$ are known as soon as $\widehat \sigma_0$ and $\widehat\sigma_1$ are given. The proposition below states the conditions that the flow rule should satisfy in order to be compatible with the stress laws and the Helmholtz energy.

\begin{proposition} \label{sigma0*}
Under the above assumptions, condition \eqref{condition 3} holds if and only if 
\begin{equation}\label{on fait ce qu'on peut}
    \bigl\|\widehat\sigma_0^*(\theta,\xi)\bigl\|^2\le 4\rho\Bigl(2\widehat\mu_1(\theta,\xi)\widehat\eta_0(\theta,\xi)-\rho\bigl(\widehat\lambda_0(\theta,\xi)-\widehat\eta_1(\theta,\xi)\bigr)^2\Bigr)\Bigl\|\frac{\partial\widehat a_m}{\partial \xi}(\theta,\xi)\Bigr\|^2,
\end{equation}
for all $\theta,\xi$. 
The condition
\begin{equation}\label{on fait ce qu'on peut 2}
    \rho\bigr(\widehat\lambda_0(\theta,\xi)-\widehat\eta_1(\theta,\xi)\bigl)^2\le 2\widehat\mu_1(\theta,\xi)\widehat\eta_0(\theta,\xi),
\end{equation}
when $\frac{\partial\widehat a_m}{\partial\xi}(\theta,\xi)\neq0$ is then necessary.
\end{proposition}

\begin{proof}
Let us replace all of our above constitutive assumptions into \eqref{condition 3}. This yields 
$$
    \biggl(\Bigl(2\rho(\widehat\lambda_0(\theta,\xi)-\widehat\eta_1(\theta,\xi))\frac{\partial\widehat a_m}{\partial \xi}(\theta,\xi)+\widehat\sigma_0^*(\theta,\xi)\Bigr):\bar d\biggr)^2
    \le8\rho\widehat\mu_1(\theta,\xi)\widehat\eta_0(\theta,\xi)\Bigl\|\frac{\partial\widehat a_m}{\partial\xi}\Bigr\|^2,
$$
for all $\bar d$ on the unit sphere.
The right-hand side does not depend on $\bar d$ and the maximum of the left-hand side is equal to $ \bigl\|2\rho(\widehat\lambda_0-\widehat\eta_1)\frac{\partial\widehat a_m}{\partial \xi}+\widehat\sigma_0^*\bigr\|^2=\bigl\|2\rho(\widehat\lambda_0-\widehat\eta_1)\frac{\partial\widehat a_m}{\partial \xi}\bigr\|^2+\bigl\|\widehat\sigma_0^*\bigr\|^2$, since $\frac{\partial\widehat a_m}{\partial \xi}$ and $\widehat\sigma_0^*$ are orthogonal.
Hence the result.
\end{proof}

\begin{remark}
When $\frac{\partial\widehat a_m}{\partial\xi}(\theta,\xi)=0$,  condition \eqref{on fait ce qu'on peut} is equivalent to $\dev(\widehat\sigma_0(\theta,\xi))=0$, which is already clear from \eqref{condition 3}. On the other hand,  when $\frac{\partial\widehat a_m}{\partial\xi}(\theta,\xi)\neq0$, we can write $\widehat\sigma_0^*=\omega^*\varsigma^*$ where $\varsigma^*\in \vect(\xi,\dev(\xi^2))$ is such that $\varsigma^*:\frac{\partial\widehat a_m}{\partial\xi}=0$ and $\|\varsigma^*\|=\bigl\|\frac{\partial\widehat a_m}{\partial\xi}\bigr\|$, with $\omega^*=0$ when $\xi$ is uniaxial. 
Using this representation, inequality \eqref{on fait ce qu'on peut} becomes
\begin{equation}\label{on fait ce qu'on peut 3}
   \omega^*(\theta,\xi)^2\le 4\rho\Bigl(2\widehat\mu_1(\theta,\xi)\widehat\eta_0(\theta,\xi)-\rho\bigl(\widehat\lambda_0(\theta,\xi)-\widehat\eta_1(\theta,\xi)\bigr)^2\Bigr).
\end{equation}
Note that we can always find $\widehat\eta_0$ nonnegative and $\widehat\eta_1$ such that \eqref{on fait ce qu'on peut} is satisfied.

An interesting particular case is $\widehat\sigma^*_0=0$, in which case the necessary and sufficient condition \eqref{on fait ce qu'on peut} simplifies as \eqref{on fait ce qu'on peut 2} when $\frac{\partial\widehat a_m}{\partial\xi}(\theta,\xi)\neq0$, which is thus also sufficient in that case. 
\end{remark}

\begin{remark}
It is possible to be slightly more general concerning the choice of $\widehat k_1$, while keeping the hypothesis $\widehat\sigma_1(\theta,\xi)[d]=2\widehat\mu_1(\theta,\xi)d$. 
For any $\widehat k_1$ given by \eqref{isotropie 3}, condition \eqref{condition 3} reads
\begin{equation}
    \Bigl\| \dev \widehat\sigma_0(\theta,\xi) -\rho\widehat k_1^T(\theta,\xi)\frac{\partial\widehat a_m}{\partial\xi}(\theta,\xi)\Bigr\|^2\le 8\rho\widehat\mu_1(\theta,\xi)\widehat\eta_0(\theta,\xi) \Bigl\|\frac{\partial\widehat a_m}{\partial\xi}(\theta,\xi)\Bigr\|^2,
\end{equation}
where for any $(\theta,\xi)$, $\widehat k_1^T\!(\theta,\xi)$ is the adjoint of $\widehat k_1(\theta,\xi)$. It is readily seen that, for all $d'\in V$,
\begin{multline}
    \widehat k_1^T\!(\theta,\xi) [d']=\biggl(\Bigl(\sum_{j=1}^2\widehat\delta'_j\dev(\xi^j)\Bigr){:}\,d'\biggr)\xi
    +\biggl(\Bigl(\sum_{j=1}^2\widehat\delta''_j\dev(\xi^j)\Bigr){:}\,d'\biggr)\!\dev(\xi^2)\\+\sum_{j=0}^3\widehat\delta'''_j\dev(\xi^jd'{+}d'\xi^j).
\end{multline} 
The sum of the first two terms belongs to 
$\vect(\xi,\dev(\xi^2))$. When $d'=\frac{\partial\widehat a_m}{\partial\xi}(\theta,\xi)$, the third term is in $\vect(\xi,\dev(\xi^2))$ as well due to the Cayley-Hamilton theorem. In the same vein as above, we can write
\begin{multline*}
    \dev \widehat\sigma_0(\theta,\xi) -\rho\widehat k_1^T(\theta,\xi)\frac{\partial\widehat a_m}{\partial\xi}(\theta,\xi)
    = \widehat \lambda_1(\theta,\xi) \frac{\partial\widehat a_m}{\partial\xi}(\theta,\xi) \\+ \Bigl(\dev \widehat\sigma_0(\theta,\xi) -\rho\widehat k_1^T(\theta,\xi)\frac{\partial\widehat a_m}{\partial\xi}(\theta,\xi)\Bigl)^*,
\end{multline*}
where the last term can be explicitly computed in $\vect(\xi,\dev(\xi^2))$.
Condition \eqref{condition 3} now reads
\begin{equation*}
    \Bigl\| \widehat\sigma_0^*(\theta,\xi) -\rho\Bigl(\widehat k_1^T(\theta,\xi)\frac{\partial\widehat a_m}{\partial\xi}(\theta,\xi)\Bigl)^*\Bigr\|^2
    \leq \Bigl(8\rho\widehat\mu_1(\theta,\xi)\widehat\eta_0(\theta,\xi) - \widehat \lambda_1^2(\theta,\xi)\Bigr) \Bigl\|\frac{\partial\widehat a_m}{\partial \xi}(\theta,\xi)\Bigr\|^2.
\end{equation*}

The result in Proposition \ref{sigma0*} had all $\widehat\delta_j'$, $\widehat\delta_j''$, $\widehat\delta_j'''$ equal to $0$ except $\widehat\delta_0'''=\widehat\eta_1$.
\end{remark}

\begin{example}
To sum things up, in the case of the de Gennes free energy  \eqref{bulk dG}, 
we have
$$
    \widehat a_m(\theta,\xi)=\alpha(\theta-\theta^*)\tr(\xi^2)+\frac c2\bigl(\tr(\xi^2)\bigr)^2-\frac{2b}{3}\tr(\xi^3),
$$
and 
$$
    \frac{\partial\widehat a_m}{\partial\xi}(\theta,\xi)=2\bigl(\alpha(\theta-\theta^*)+c\tr(\xi^2)\bigr)\xi-2b\dev(\xi^2),
$$
which, together with the previous hypotheses on $\widehat\sigma_1$, $\widehat k_0$, and $\widehat k_1$,  and conditions \eqref{on fait ce qu'on peut} or \eqref{on fait ce qu'on peut 3}, plus 
condition \eqref{condition thermique} on $\widehat\kappa$,
yields a family of MacMillan internal variable models compatible with the de Gennes energy, that is guaranteed to satisfy the second law of thermodynamics. 
A possible simple choice for the stress tensor is
$$
    \widehat\sigma(\theta,d,\xi)=4\rho\widehat\lambda_0(\theta,\xi)\bigl((\alpha(\theta-\theta^*)+c\tr(\xi^2))\xi-b\dev(\xi^2)\bigr)+2\widehat\mu_1(\theta,\xi)d,
$$
and for the flow rule 
$$
    \widehat k_s(\theta,\xi)= -2\widehat\eta_0(\theta,\xi)\bigl((\alpha(\theta-\theta^*)+c\tr(\xi^2))\xi-b\dev(\xi^2)\bigr)
    +2\widehat\eta_1(\theta,\xi)d.
$$
with $\widehat\mu_1\ge 0$, $\widehat\eta_0\ge0$ and $\rho\bigr(\widehat\lambda_0(\theta,\xi)-\widehat\eta_1(\theta,\xi)\bigl)^2\le 2\widehat\mu_1(\theta,\xi)\widehat\eta_0(\theta,\xi)$.

\end{example}

Let us close this section with two remarks in particular cases. 
\begin{remark}
The first remark concerns the uniaxial case, under some of the above hypotheses. Let us assume that the fluid is at rest $v=0$ and in thermal equilibrium $\theta(x,t)=\theta_0(x)$. The evolution of the order tensor $\xi$ is then decoupled from the dynamics and heat equations, and the former equilibrium hypotheses can be ensured in principle by adjusting the body forces and heat sources as usual. We assume \eqref{flowruleMcM} in conjunction with \eqref{forme de k0}. We show that, in this case, if $\xi$ is uniaxial at some time $t_0$, then it is uniaxial for all $t$, for any Helmholtz free energy, without any reference to the validity of the second law.

First of all, $\frac{\partial\widehat a_m}{\partial\xi}(\theta,0)=0$ so that $\xi=0$ is an equilibrium point for the evolution of $\xi$. It is thus enough to consider $\xi(t_0)=\xi_0=s_0\bigl(n_0\otimes n_0-\frac13I\bigr)$ for some unit vector $n_0$ and scalar $s_0\neq0$.

Let us generically rewrite the evolution equation for $\xi$ at some point $x$ fixed as,
$$
    \frac{\partial\xi}{\partial t}=f_1(\xi)\xi+f_2(\xi)\dev(\xi^2),\quad \xi(t_0)=\xi_0,
$$
with $f_1,f_2$ scalar-valued, cf. \eqref{forme de da/dxi} and \eqref{forme de k0}, using the fact that $v=0$ and $\theta$ is constant in time and independent of $\xi$.  

Now if $\xi$ is uniaxial, then $\dev(\xi^2)=\frac s3\xi$. Moreover, we can recover $s$ from $\xi$ with the formula $s=3\sqrt[3]{\det(\xi)/2}$.
Let us set $\ell(\xi)=\bigl(f_1(\xi)+\sqrt[3]{\det(\xi)/2}f_2(\xi)\bigr)$ and consider the auxiliary Cauchy problem
$$
    \frac{\partial\xi}{\partial t}=\ell(\xi)\xi,\quad \xi(t_0)=\xi_0.
$$
This Cauchy problem satisfies the conditions of the Picard-Lindelöf theorem. It thus has a unique maximal solution, which is obviously of the form $\xi(t)=\frac{s(t)}{s_0}\xi_0$ with $s(t_0)=s_0$, by local uniqueness, and thus remains uniaxial for all $t$. It is consequently the solution of the original Cauchy problem, also by local uniqueness.
\end{remark}

\begin{remark}
The second remark has to do with the time evolution of the free energy and free energy minimization. In works dealing with the static case, much of the focus is indeed on minimizing the free energy, see \cite{JMB}. Let us assume again that $v=0$ and $\theta(x,t)=\theta_0(x)$. In this case, $\frac{\partial a_m}{\partial t}(x,t)=\dot a_m(x,t)=-\frac{d_{\textnormal{int}}(x,t)}{\rho}+\sigma(x,t):d(x,t)-\dot\theta(x,t) s_m(x,t)=-\frac{d_{\textnormal{int}}(x,t)}{\rho}\le 0$ and the free energy density at point $x$ fixed is nonincreasing in time because of the second law, without any additional hypothesis.
Since $a_m(x,t)=\widehat a_m(\theta_0(x),\xi(x,t))$, if 
$\widehat a_m$ is coercive with respect to $\xi$, as is the case for the de Gennes energy, then it is a Lyapunov function for the evolution of $\xi$. Therefore, we have global existence in time and boundedness for $\xi(x,t)$. Naturally, this still does not mean that $\xi(x,t)$ must converge to a minimizer of $\widehat a_m(\theta_0(x),\xi)$ when $t\to+\infty$, even in this simple static case.
\end{remark}

\section{A generalization of MacMillan's model}\label{Section trois pas en avant}
Let us go one step beyond MacMillan's model in the same direction. We assume constitutive laws of the form 
\begin{equation}\label{sigma McM plus}
     \widehat\sigma(h,\theta,\xi)=\widehat\sigma_0(\theta,\xi)+\widehat \sigma_1(\theta,\xi)[d]+\widehat\sigma_2(\theta,\xi)[d^2],
\end{equation}
and
\begin{equation}\label{flowrule McM plus}
    \widehat k_s(d,\theta,\xi)=\widehat k_0(\theta,\xi)+\widehat k_1(\theta,\xi)[d]+\widehat k_2(\theta,\xi)[d^2],
\end{equation}
where $\widehat\sigma_0$ and $\widehat\sigma_1$ are as in \eqref{isotropie 0} and \eqref{isotropie 1}, $\widehat k_0$  and $\widehat k_1$ as in \eqref{isotropie 2}
and \eqref{isotropie 3}, and $\widehat\sigma_2$ and $\widehat k_2$ are additional objective functions with values in $\mathcal{L}(\Sym_3;\Sym_3)$ and $\mathcal{L}(\Sym_3;V)$ respectively. In view of formula \eqref{stress R-E} and for simplicity instead of going for full generality, we take the quadratic stress part as follows,
\begin{equation}\label{sigma 2, indices 3 7+}
    \widehat\sigma_2(\theta,\xi)[d^2]=\widehat\alpha_2(\theta,\xi)\tr(d^2)\xi+\widehat\alpha_3(\theta,\xi)d^2+\widehat\alpha_7(\theta,\xi)(\xi d^2+d^2\xi).
\end{equation}

In general, we can also expect additional terms involving $\tr(d^2)$, $\tr(\xi d^2)$, and $\tr(\xi^2d^2)$ multiplied by objective tensor-valued functions of $\xi$.

\begin{proposition}\label{plus fort que McM...+}
A set of necessary and sufficient conditions for the extended model to satisfy the mechanical/internal Clausius-Planck inequality is 
\begin{gather}
    \widehat\sigma_2=0,\label{condition 0 ++}\\
    \frac{\partial\widehat a_m}{\partial\xi}(\theta,\xi):\widehat k_0(\theta,\xi)\le0,
    \label{condition 1 ++}\\
    \bigl(\widehat\sigma_1(\theta,\xi)[\bar d\,]\bigr):\bar d-\rho\frac{\partial\widehat a_m}{\partial\xi}(\theta,\xi):\bigl(\widehat k_2(\theta,\xi)[\bar d{\,}^2]\bigr)\ge0,
    \label{condition 2 ++}\\
    \Bigl(\widehat\sigma_0(\theta,\xi):\bar d-\rho\frac{\partial\widehat a_m}{\partial\xi}(\theta,\xi):\bigl(\widehat k_1(\theta,\xi)[\bar d\,]\bigr)\Bigr)^2
    \hskip 7.5cm\nonumber\\
    {}+4\rho\Bigl(\frac{\partial\widehat a_m}{\partial\xi}(\theta,\xi):\widehat k_0(\theta,\xi)\Bigr)\Bigl(\bigl(\widehat \sigma_1(\theta,\xi)[\bar d\,]\bigr):\bar d-\rho\frac{\partial\widehat a_m}{\partial\xi}(\theta,\xi):\bigl(\widehat k_2(\theta,\xi)[\bar d{\,}^2]\bigr)\Bigr)\le0.
    \label{condition 3 ++}
\end{gather}
for all $\bar d\in V$, $\|\bar d\|=1$, $\theta\in \R_+^*$, $\xi\in V$.
\end{proposition}

\begin{proof}
We rewrite the  Clausius-Planck inequality for the extended model for the reader's convenience, 
\begin{equation}\label{CP ordre 2+}
    \bigl(\widehat\sigma_0+ \widehat\sigma_1[d]+\widehat\sigma_2[d^2]\bigr) :d -\rho\frac{\partial\widehat a_m}{\partial\xi}:\bigl(\widehat k_0+\widehat k_1[d]+ \widehat k_2[d^2]\bigr)\ge0.
\end{equation}
Replacing $d=\lambda\bar d$, we obtain
\begin{multline}
    -\rho\frac{\partial\widehat a_m}{\partial\xi}:\widehat k_0+\lambda\Bigl(\widehat\sigma_0:\bar d{-}\rho\frac{\partial\widehat a_m}{\partial\xi}:\widehat k_1[\bar d\,]\Bigr)+\lambda^2\Bigl(\bigl(\widehat \sigma_1[\bar d\,]\bigr):\bar d{-}\rho\frac{\partial\widehat a_m}{\partial\xi}:\bigl(\widehat k_2[\bar d{\,}^2]\bigr)\Bigr)\\+\lambda^3\widehat\sigma_2[\bar d{\,}^2]:\bar d\ge 0.
\end{multline}
Letting $\lambda\to\pm\infty$, we see that we must have $\widehat\sigma_2[\bar d{\,}^2]:\bar d=0$ for all $\theta$, $\xi$, and $\bar d$, so that we can dispense with the $\|\bar d\|=1$ constraint, and write (without $\theta$)
$$
    \widehat\alpha_2(\xi)\tr(d^2)\tr(\xi d)+\widehat\alpha_3(\xi)\tr(d^3)+2\widehat\alpha_7(\xi)\tr(\xi d^3)=0
$$
for all $\xi$ and $d$. Let us discuss this equality according to $\xi=\sum_{i=1}^3\lambda_i e_i\otimes e_i$, with $\sum_{i=1}^3\lambda_i=0$ and $e_i$ orthonormal.
\begin{itemize}
    \item If $\xi\neq0$, then it has at least two distinct eigenvalues, say $\lambda_1\neq\lambda_2$ without loss of generality. We first take $d_1=e_1\otimes e_1-e_2\otimes e_2=d_1^3$ so that $\tr(d_1^3)=0$ and $\tr(\xi d_1)=\tr(\xi d_1^3)=\lambda_1-\lambda_2\neq 0$. It follows that $\widehat \alpha_2(\xi)+\widehat \alpha_7(\xi)=0$. Secondly, using now $d_2=2e_1\otimes e_1-e_2\otimes e_2-e_3\otimes e_3$, we see that $\widehat \alpha_3(\xi)=0$. 
    At this point, we have $\widehat\alpha_2(\xi)\xi:\bigl(\|d\|^2d-d^3
    \bigr)=0$ for all $d$.
    
    Thirdly, we take $d=\xi$. Then, $\xi:\bigl(\|d\|^2d-d^3)= \|\xi\|^4-\tr\xi^4= \frac12 \|\xi\|^4$, so that $\widehat\alpha_2(\xi)=0$ and in fine  $\widehat \sigma_2(\xi)=0$.
    \item If $\xi=0$, using the same $d_2$ as above, we also get $\alpha_3(0)=0$ so that again $\widehat\sigma_2(0)=0$.
\end{itemize}
It follows that \eqref{CP ordre 2+} is actually quadratic in $d$ and we are back in a setting that is similar to that of Proposition \ref{les conditions pour McM},
with the term $\bigl(\widehat\sigma_1[\bar d\,]\bigr):\bar d$ replaced by the term $\bigl(\widehat \sigma_1[\bar d\,]\bigr):\bar d-\rho\frac{\partial\widehat a_m}{\partial\xi}:\bigl(\widehat k_2[\bar d{\,}^2]\bigr)$.
\end{proof}

We thus see that a stress extension of this specific form is not thermodynamically meaningful: $\widehat \sigma$ must remain linear with respect to $d$. This does not rule out other, possibly non polynomial, extensions of the stress constitutive law.

Concerning the quadratic part of the flow rule, we make an assumption analogous to \eqref{sigma 2, indices 3 7+} for the stress, namely
\begin{align}
    \widehat k_2(\theta,\xi)[d^2]&=\widehat\delta_3(\theta,\tr(\xi^2),\tr(\xi^3))\tr(d^2)\xi+\widehat\delta_4\bigl(\theta,\tr(\xi^2),\tr(\xi^3)\bigr)\dev(d^2)\nonumber\\
    &\qquad\qquad\qquad\qquad\qquad\qquad+\widehat\delta_5\bigl(\theta,\tr(\xi^2),\tr(\xi^3)\bigr)\dev(\xi d^2+d^2\xi) \label{k 2, indices 3 7+}.
\end{align}
Again, more general laws could be considered, based on formula \eqref{stress R-E}. As opposed to the stress extension, this flow rule extension actually does play a thermodynamic  role. 

Indeed, we now proceed to show that the extended model is versatile enough to produce complex behaviors that are compatible with the second law of thermodynamics. We do not 
insist on the necessary side of Proposition \ref{plus fort que McM...+}, but rather concentrate on showing directly that there exist compatible extensions, while paying attention to the regularity of the $\widehat k_2$ constructed, since the latter enter the right-hand side of a differential equation.

We first define the following auxiliary functions in the variables $(\tau_2,\tau_3)\in\mathcal{D}$ which combine the Helmholtz energy and the second degree part of the flow rule, and which will appear naturally in ensuing computations:
\begin{multline} \label{les zetas+}
    \widehat\zeta_0=\Bigl(2\tau_2\frac{\partial\widehat\psi}{\partial\tau_2}+3\tau_3\frac{\partial\widehat\psi}{\partial \tau_3}\Bigr)\widehat\delta_3 -\tau_2\frac{\partial\widehat\psi}{\partial \tau_3}\widehat\delta_4+2\tau_3\frac{\partial\widehat\psi}{\partial\tau_3}\widehat\delta_5,\\
    \widehat\zeta_{1}=2\frac{\partial\widehat\psi}{\partial\tau_2}\widehat\delta_4+{\tau_2}\frac{\partial\widehat\psi}{\partial\tau_3}\widehat\delta_5,\quad
    \widehat\zeta_{2}=3\frac{\partial\widehat\psi}{\partial\tau_3}\widehat\delta_4+4\frac{\partial\widehat\psi}{\partial\tau_2}\widehat\delta_5.
\end{multline}

\begin{proposition}\label{bon ok ça marche}
Assume that $\widehat\sigma_0$, $\widehat\sigma_1$, $\widehat k_0$ and $\widehat k_1$ are such that inequalities \eqref{condition 1}, \eqref{condition 2} and \eqref{condition 3} are satisfied. Then there exist functions $\widehat\delta_3$, $\widehat\delta_4$ and $\widehat\delta_5$ that define locally Lipschitz functions of $\xi$ and such that the Clausius-Planck inequality \eqref{CP ordre 2+} is satisfied.
\end{proposition}

\begin{proof}
We know that $\widehat\sigma_2=0$. 
The Clausius-Planck inequality may be rewritten as 
$$
    \bigl(\widehat\sigma_0+ \widehat\sigma_1[d]\bigr) :d -\rho\frac{\partial\widehat a_m}{\partial\xi}:\bigl(\widehat k_0+\widehat k_1[d]\bigr)-\rho\frac{\partial\widehat a_m}{\partial\xi}:\widehat k_2[d^2]\ge0.
$$
Under the above hypotheses, we have previously shown that 
$$
    \bigl(\widehat\sigma_0+ \widehat\sigma_1[d]\bigr) :d -\rho\frac{\partial\widehat a_m}{\partial\xi}:\bigl(\widehat k_0+\widehat k_1[d]\bigr)\ge0.
$$
It is thus enough for our purposes to independently ensure that
\begin{equation}\label{bon ca suffit}
\frac{\partial\widehat a_m}{\partial\xi}:\widehat k_2[d^2]\le0.
\end{equation}
Let us compute the expression in the left-hand side of \eqref{bon ca suffit}.
Using the representation formula \eqref{forme de da/dxi} and the chosen form \eqref{k 2, indices 3 7+} for $\widehat k_2$, we get 
\begin{multline*}
    \frac{\partial\widehat a_m}{\partial\xi}(\xi):\bigl(\widehat k_2(\xi) [d^2]\bigr)= \Bigl(2\frac{\partial\widehat\psi}{\partial\tau_2}\tr(\xi^2)+3\frac{\partial\widehat\psi}{\partial\tau_3}\dev(\xi^2):\xi\Bigr)\widehat\delta_3\|d\|^2\\
    \hskip25pt+2\frac{\partial\widehat\psi}{\partial\tau_2}\widehat\delta_4\bigl(\xi:\dev(d^2)\bigr)+3\frac{\partial\widehat\psi}{\partial\tau_3}\widehat\delta_4 \bigl(\dev(\xi^2):\dev(d^2)\bigr)\\
    +2\frac{\partial\widehat\psi}{\partial\tau_2}\widehat\delta_5\bigl(\xi :\dev(\xi d^2+d^2\xi)\bigr)+3\frac{\partial\widehat\psi}{\partial\tau_3}\widehat\delta_5 \bigl(\dev(\xi^2):\dev(\xi d^2+d^2\xi)\bigr),
\end{multline*}
where, for the sake of brevity, we omitted all arguments $\tau_2, \tau_3$ in the scalar functions $\frac{\partial\widehat\psi}{\partial\tau_j}$ and $\widehat\delta_j$. Since as a rule, $\dev(A):\dev(B)=A:B-\frac13\tr(A)\tr(B)$, there holds
\begin{multline*}
    \dev(\xi^2):\dev(d^2)=\|\xi d\|^2-\frac 1 3\tr(\xi^2)\|d\|^2\text{ and }\\
    \dev(\xi^2):\dev(\xi d^2)=\frac 1 6\big(\tr(\xi^2){\,}\xi d:d +2\tr(\xi^3)\|d\|^2\bigr),
\end{multline*}
using for the latter the fact that $\xi^3=\frac12\tr(\xi^2)\xi+\frac13\tr(\xi^3)I$, which follows from the Cayley-Hamilton theorem and from the fact that $\tr(\xi)=0$. We thus obtain
\begin{align*}
    \frac{\partial\widehat a_m}{\partial\xi}(\xi):\bigl(\widehat k_2(\xi) [d^2]\bigr)&=\Bigl(2\frac{\partial\widehat\psi}{\partial\tau_2}\tr(\xi^2)+3\frac{\partial\widehat\psi}{\partial\tau_3}\tr(\xi^3)\Bigr)\widehat\delta_3\|d\|^2\\
    &\quad+2\frac{\partial\widehat\psi}{\partial\tau_2}\widehat\delta_4{\,}\xi d:d
    +3\frac{\partial\widehat\psi}{\partial\tau_3}\widehat\delta_4{\,}\bigl(\|\xi d\|^2-\frac 1 3\tr(\xi^2)\|d\|^2\bigr)\\
    &\qquad+4\frac{\partial\widehat\psi}{\partial\tau_2}\widehat\delta_5{\,}\|\xi d\|^2+\frac{\partial\widehat\psi}{\partial\tau_3}\widehat\delta_5{\,}\bigl(\tr(\xi^2){\,}\xi d:d+2\tr(\xi^3)\|d\|^2\bigr)\\
    &=\biggl(\Bigl(2\tau_2\frac{\partial\widehat\psi}{\partial\tau_2}+3\tau_3\frac{\partial\widehat\psi}{\partial\tau_3}\Bigr)\widehat\delta_3-\tau_2\frac{\partial\widehat\psi}{\partial\tau_3}\widehat\delta_4+2\tau_3\frac{\partial\widehat\psi}{\partial\tau_3}
    \widehat\delta_5\biggr)\|d\|^2\\
    &\quad+\Bigl(2\frac{\partial\widehat\psi}{\partial\tau_2}\widehat\delta_4+\tau_2\frac{\partial\widehat\psi}{\partial\tau_3}\widehat\delta_5\Bigr){\,}\xi d:d
    +\Bigl(3\frac{\partial\widehat\psi}{\partial\tau_3}\widehat\delta_4+4\frac{\partial\widehat\psi}{\partial\tau_2}\widehat\delta_5\Bigr)\|\xi d\|^2,
\end{align*}
so that
\begin{equation} \label{am kdeux+}
    \frac{\partial\widehat a_m}{\partial\xi}(\xi):\bigl(\widehat k_2(\xi)[d^2]\bigr)=
    \widehat\zeta_0{\,}\|d\|^2 +\widehat\zeta_1{\,}\xi d:d +\widehat\zeta_2{\,}\|\xi d\|^2.
\end{equation}

We are back in the familiar territory of Lemma \ref{un poil moins trivial}, and sufficient conditions for \eqref{bon ca suffit} to hold are  
\begin{equation} \label{zetas premier choix+}
  \widehat \zeta_0\le0,\;\widehat\zeta_2\le0,\text{ and }
    \widehat\zeta_1^2\le 4 \widehat\zeta_0\widehat\zeta_2,
\end{equation} 
or
\begin{equation} \label{zetas deuxième choix+}
    \Bigl(\tau_2\widehat\zeta_1^2+\frac{\tau_2^2}{2}\widehat\zeta_2^2+2\tau_3\widehat\zeta_1\widehat\zeta_2\Bigr)^{1/2}\le-\widehat\zeta_0.
\end{equation}

We can use either condition, but the second condition \eqref{zetas deuxième choix+} is slightly easier to manage.  We need to show that we can construct  functions  $\widehat\delta_3(\tau_2,\tau_3)$, $\widehat\delta_4(\tau_2,\tau_3)$ and $\widehat\delta_5(\tau_2,\tau_3)$ on  $\mathcal{D}$, cf. Lemma \ref{domaine}, that give rise to locally Lipschitz functions of $\xi$ and for which inequality \eqref{zetas deuxième choix+} holds.

Let $\widehat\psi_2^*= 2\tau_2\frac{\partial\widehat\psi}{\partial\tau_2} + 3 \tau_3 \frac{\partial\widehat\psi}{\partial\tau_3}$ as in \eqref{psi*deux}. We observe that 
$$
\widehat\zeta_0=\widehat\psi_2^*\widehat\delta_3+\bigl(-{\tau_2}\widehat\delta_4+2\tau_3\widehat\delta_5\bigr)\frac{\partial\widehat\psi}{\partial \tau_3},
$$
and that this is the only place where $\widehat\delta_3$ intervenes. So our strategy is to take $\widehat\delta_4=(\widehat\psi_2^*)^2\widehat\delta'_4$ and $\widehat\delta_5=(\widehat\psi_2^*)^2\widehat\delta'_5$ with $\widehat\delta'_4,\widehat\delta'_5$ smooth, so that their associated functions of $\xi$ are also smooth, and arbitrary. Thus $\widehat\zeta_0=\widehat\psi_2^*\bigl(\widehat\delta_3-\widehat\zeta_0'\bigr)$ with $\widehat\zeta_0'$ smooth. Inequality~\eqref{zetas deuxième choix+} is therefore implied by the equality
\begin{equation} \label{zetas deuxième choix finale}
    (\widehat\psi_2^*)^2\Bigl(\tau_2(\widehat\zeta'_1)^2+\frac{\tau_2^2}{2}(\widehat\zeta'_2)^2+2\tau_3\widehat\zeta'_1\widehat\zeta'_2\Bigr)^{1/2}=-\widehat\psi_2^*\bigl(\widehat\delta_3-\widehat\zeta_0'\bigr),
\end{equation}
where
$\widehat\zeta'_{1}=2\frac{\partial\widehat\psi}{\partial \tau_2}\widehat\delta'_4+\tau_2 \frac{\partial\widehat\psi}{\partial \tau_3}\widehat\delta'_5,\,
\widehat\zeta'_{2}=3\frac{\partial\widehat\psi}{\partial \tau_3}\widehat\delta'_4+4\frac{\partial\widehat\psi}{\partial\tau_2}\widehat\delta'_5$.
We thus see that we can take
\begin{equation} \label{delta trois final}
    \widehat\delta_3=\widehat\zeta_0'-\widehat\psi_2^*\Bigl(\tau_2(\widehat\zeta'_1)^2+\frac{\tau_2^2}{2}(\widehat\zeta'_2)^2+2\tau_3\widehat\zeta'_1\widehat\zeta'_2\Bigr)^{1/2},
\end{equation}
and satisfy \eqref{zetas deuxième choix finale}, thus consequently \eqref{zetas deuxième choix+}. Moreover, the mapping $\xi\mapsto\break \widehat\delta_3(\tr(\xi^2),\tr(\xi^3))$  is obviously locally Lipschitz on $V$, since it can easily be rewritten as 
$$
    \xi\mapsto\widehat\zeta'_0(\tr(\xi^2),\tr(\xi^3))-\widehat\psi^*_2(\tr(\xi^2),\tr(\xi^3))\|\widehat\zeta'_1(\tr(\xi^2),\tr(\xi^3))\xi+\widehat\zeta'_2(\tr(\xi^2),\tr(\xi^3))\xi^2\|,
$$
see Lemma \ref{un poil moins trivial}.
\end{proof}

\begin{remark}
In the above construction, $\widehat\delta_4$ and $\widehat\delta_5$ are actually smooth, and so is $\widehat\delta_3$ 
%{\annie $\widehat\delta_3$}
outside of the set $\tau_2(\widehat\zeta'_1)^2+\frac{\tau_2^2}{2}(\widehat\zeta'_2)^2+2\tau_3\widehat\zeta'_1\widehat\zeta'_2=0$.  
Paying a little more attention to the choice of $\widehat\delta'_4,\widehat\delta'_5$, we can make sure that $\widehat\zeta'_2$ does not vanish for $\xi\neq0$, which is possible if we assume for instance that $\nabla\widehat\psi\neq0$ on $\mathcal{D}\setminus\{(0,0)\}$. This implies that $\|\widehat\zeta'_1\xi+\widehat\zeta'_2\xi^2\|=0$ if and only if $\xi=0$. With this proviso, we see that $\widehat\delta_3$ is a smooth function on $\mathcal{D}\setminus\{(0,0)\}$. More generally, we can always manage to obtain a $\widehat k_2$ that is also smooth at $\xi=0$.
\end{remark}
\begin{remark}
 In the proof, we satisfied an inequality by satisfying the corresponding equality. It is clear that there are infinitely many more ways of achieving this, by adding to the proposed $\widehat\delta_3$ any quantity  that decreases the proposed $\widehat\zeta_0$.
\end{remark}
\begin{remark}
 In the case of the de Gennes energy, the above analysis applies since $\frac{\partial\widehat\psi}{\partial\tau_3}(\tau_2,\tau_3)=\nobreak-\frac{2b}{3}$, so that $\nabla\widehat\psi$ never vanishes.  Note that the zero locus of $\widehat\psi^*_2$ in $(\tau_2,\tau_3)$ space is the parabola $b\tau_3=\alpha(\theta-\theta^*)\tau_2+c\tau_2^2$, which may or may not intersect $\mathcal{D}\setminus\{(0,0)\}$, depending on the values of the temperature $\theta$ and of the material constants $\theta^*, \alpha, b$ and $c$.
 \end{remark}
\begin{example}
 As a kind of simplest possible nonzero example, let us take $\widehat\delta_4'=1$ and $\widehat\delta_5'=0$. In the de Gennes case, we have
$$
    \widehat\psi^*_2(\theta,\tau_2,\tau_3)=2\bigl(\alpha(\theta-\theta^*)\tau_2+c\tau_2^2-b\tau_3\bigr).
$$
The above choice thus leads to 
$$
    \widehat\delta_4=(\psi_2^*)^2=4\bigl(\alpha(\theta-\theta^*)\|\xi\|^2+c\|\xi\|^4-b\tr(\xi^3)\bigr)^2\text{ and }\widehat\delta_5=0,
$$  when expressed as  functions of $\xi$. Continuing the computations, we obtain
$$
    \widehat\zeta_0'=-\frac{2b}3\tau_2\widehat\psi_2^*,\;\widehat\zeta_1'=2(\alpha(\theta-\theta^*)+c\tau_2),\;\widehat\zeta_2'=-2b.
$$
 This yields the following formula for $\widehat\delta_3$, expressed again as a function of $\xi$,
$$
    \widehat\delta_3=2\bigl(\alpha(\theta-\theta^*)\|\xi\|^2+c\|\xi\|^4-b\tr(\xi^3)\bigr)\Bigl(-\frac{2b}3\|\xi\|^2-\bigl\|2\bigl(\alpha(\theta-\theta^*)+c\|\xi\|^2\bigr)\xi-2b\xi^2\bigr\|\Bigr).
$$
 Note that this function is $C^1$ at $\xi=0$.
 It should be kept in mind that the resulting flow rule $\widehat k_2$ is proposed on no other physical grounds than the fact that it happens to satisfy the second law of thermodynamics, which is pretty minimal.
\end{example}

\end{document}